\documentclass[12pt]{article}
\usepackage{latexsym,float,amsfonts,amsmath,amsthm,epsfig,makeidx,tabularx,hyperref,color}
\usepackage{authblk}
\usepackage{graphicx}
\usepackage{subfigure}
\usepackage[labelfont=bf, labelsep=space]{caption}
\usepackage{tikz} 
\definecolor{darkred}{rgb}{0.5,0,0}
\definecolor{darkgreen}{rgb}{0,0.5,0}
\definecolor{darkblue}{rgb}{0,0,0.5}
\hypersetup{colorlinks=true,linkcolor=blue,citecolor=red,urlcolor=blue}
\input{Qcircuit.tex}
\newtheorem{theorem}{Theorem}

\newtheorem{lem}{Lemma}
\newtheorem{prop}{Proposition}
\newtheorem{cor}{Corollary}

\newcommand\e{\varepsilon}
\newcommand\la{\lambda}

\newcommand\al{\alpha}
\newcommand\ka{\kappa}

\newcommand{\sinc}{{\rm sinc}}
\newcommand{\norm}[1]{\| #1\|}
\newcommand{\inner}[2]{\langle #1 | #2 \rangle}

\title{Estimating the ground state energy of the Schr\"odinger equation
for convex potentials}

\author[1]{Anargyros Papageorgiou
  \thanks{Electronic address: \texttt{ap@cs.columbia.edu}}}
\author[1,2]{Iasonas Petras
  \thanks{Electronic address: \texttt{ipetras@cs.columbia.edu}}}
\affil[1]{Department of Computer Science, Columbia University}
\affil[2]{Department of Computer Science, Princeton University}

\date{\today }

\begin{document}

\maketitle

\begin{abstract}
In 2011, the fundamental gap conjecture for Schr\"odinger operators was proven.
This can be used to estimate the ground state energy of the time-independent Schr\"odinger equation with
a convex potential and relative error
$\e$. Classical deterministic algorithms solving this problem have cost exponential in the number of its degrees
of freedom $d$.
We show a quantum algorithm, that is based on a perturbation method, for estimating the ground state
energy with relative error $\e$. 
The cost of the algorithm  is polynomial in $d$ and $\e^{-1}$, while
the number of qubits is polynomial in $d$ and $\log\e^{-1}$.
In addition, we present an algorithm for preparing a quantum
state that overlaps within $1-\delta$, $\delta \in (0,1)$, with the ground state eigenvector of the discretized Hamiltonian. 
This algorithm also approximates the ground state with relative 
error $\e$. 
The cost of the algorithm is polynomial in $d$, $\e^{-1}$ and $\delta^{-1}$,
while the number of qubits is polynomial in $d$, $\log\e^{-1}$ and $\log\delta^{-1}$. 


\end{abstract}

\section{Introduction}

The power of quantum computers 
has been studied extensively. In some cases the results have been very exciting and encouraging while in others the results show
limitations of quantum computation.
Shor's quantum algorithm for 
computing prime factors of a number \cite{Shor97} 
has provided an exponential speed-up compared to the fastest classical
algorithm known, the number field sieve.  
Similarly, Grover's algorithm for searching an unstructured database provides a quadratic speed-up 
compared to the best classical algorithm \cite{Grover97}. 
There are also results demonstrating that certain problems are very hard for quantum computers.
For example, a decision problem about the ground state energy of a local Hamiltonian is QMA-complete \cite{KKR06}.

One of the most important problems in computational sciences is to calculate the properties of physical 
and chemical systems. Such systems are described by the Schr\"odinger equation. 
Typically, this equation imposes significant computational demands in carrying out precise
calculations \cite{Lub}. 

Of particular interest is the estimation of the ground state energy
of the time independent Schr\"odinger equation. In particular,
consider the eigenvalue problem 
\begin{eqnarray}
\left(-\frac{1}{2}\Delta + V \right) \Psi(x) &=& E \Psi(x) \quad x \in I_d = (0,1)^d, 
\label{TISE1} \\ 
\Psi(x) &=& 0 \quad x \in \partial I_d,
\label{TISE2}
\end{eqnarray}
where $\Psi$ is a normalized eigenfunction.
Assume
that all the masses and the normalized Planck constant are one, and that
the potential $V$ is a smooth function as we will explain below.
Our problem is to estimate  the ground state energy (i.e., the smallest eigenvalue), $E_0$, with error $\e$. 

Such eigenvalue problems can be solved by discretizing the continuous Hamiltonian to obtain a symmetric
matrix, and then by
approximating the smallest matrix eigenvalue. Eigenvalue problems involving 
symmetric matrices are conceptually easy and methods such as the bisection method can be
used to solve them with cost proportional to the matrix size, modulo polylog factors. The
difficulty is that the discretization leads to a matrix of size that is
exponential in $d$.
Hence, the cost for approximating the matrix eigenvalue is prohibitive when $d$ is large. 
In fact,
a stronger result is known, namely the cost  
of any deterministic classical algorithm 
must be at least exponential in $d$, i.e., the problem suffers from the 
curse of dimensionality \cite{Pap07}.

In certain cases, quantum algorithms may be able to break the curse of dimensionality by
computing $\e$-accurate eigenvalue estimates with cost polynomial in $\e^{-1}$ and in the degrees of
freedom $d$.
This was shown in \cite{PPTZ12}
where we see that if the potential is smooth, nonnegative and uniformly bounded by a relatively small
constant
there exists a quantum algorithm approximating the ground state energy with relative error 
$\e$ and cost proportional to $d \e^{-(3+\eta)}$, where 
$\eta>0$ is arbitrary.

It is interesting 
to investigate conditions for $V$ beyond those of 
\cite{Pap07,PPTZ12} where quantum algorithms, possibly ones implementing perturbation methods,
approximate the ground state energy without suffering from the curse of dimensionality. 
Indeed, in this paper we assume that $V$ is convex and uniformly bounded by
$C>1$, \footnote{We remark that $C$ may depend on $d$. The results of this paper hold for $C=o\left(d^5\right)$, as we will see later
in Section \ref{sec:eiv}.} 
as opposed to $C\leq 1$ in \cite{Pap07,PPTZ12}. In addition, just like in \cite{PPTZ12}, the potential is non-negative
and its partial derivatives are continuous and uniformly bounded by a constant $C^\prime >0$.
We derive a quantum algorithm estimating the ground state energy with relative error $\e$
and cost polynomial in $\e^{-1}$ and $d$.


In particular our algorithm solves the eigenvalue problem 
for a sequence of Hamiltonians $H_\ell = -\frac{1}{2}\Delta + V_\ell$, for $\ell = 1,2,\ldots ,L$, 
where $V_\ell = \ell \cdot V/L$. The algorithm proceeds in $L$ stages; see Fig.~\ref{Fig:Alg_simple}.
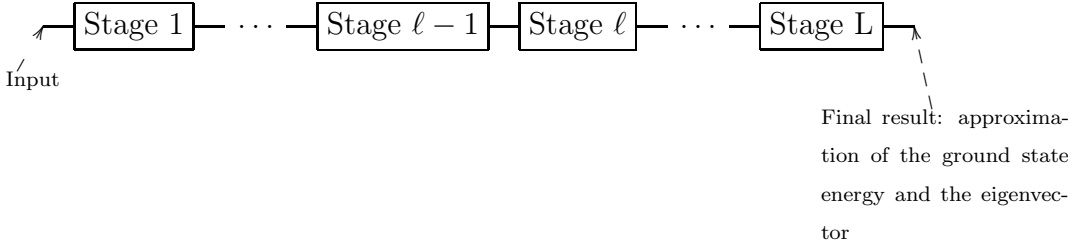
\begin{figure}[H]
\centering
\Qcircuit @R=1em @C=1em {
& & & & & & \gate{\rm Stage \,\, 1} & \qw &\dots & & \gate{\rm Stage \,\, \ell -1} & \gate{\rm Stage \,\, \ell} & \qw &  \dots &  
& \gate{\rm Stage \,\, L} & \qw & &  \\
& & & & \ar@{-->}[ru]|<{\rm\ \ \ \  Input} & & & & & \\
& & & & & \\
& & & & & \\
& & & & & & & & & & & &  & & & &  & \ar@{-->}[uuuul]|<{\textrm{\parbox[t]{0.2\linewidth}{Final result: approximation of the ground state energy and the eigenvector}}}
}
\caption{The $L$ stages of the algorithm}
\label{Fig:Alg_simple}
\end{figure} 

In each stage, the algorithm produces an approximate 
ground state eigenvector of $H_\ell$ which is passed on to the next stage. The fact that $V$ is convex 
allows us to use lower bounds on the fundamental gap \cite{AC11} and to select $L$ accordingly
so that the ground state eigenvectors of the successive Hamiltonians have a large enough 
\lq\lq overlap\rq\rq\  between them. 
This means that the (approximate) ground state eigenvector of $H_\ell$ is also an approximate 
ground state eigenvector of $H_{\ell+1}$. Our algorithm performs  a measurement at every stage,
which produces the desired outcome with a certain probability. We select the parameters of the
algorithm so that the total success probability is at least $3/4$.

In terms of the number of quantum operations and queries, the resulting cost is
\[c(k) \cdot \e^{-(3+\frac{1}{2k})} \cdot C^{\frac{4-2\eta}{1-\eta} + \frac{5-2\eta}{2k(1-\eta)}}\cdot d^{1+\frac{4-2\eta}{1-\eta}+ \frac{3}{2k(1-\eta)}} \]
and the number of qubits is
\[
3 \log \e^{-1} + \frac{2 - \eta}{1-\eta} \cdot \log (Cd) + \Theta \left(d\cdot \log \e^{-1}\right).
\]
In the expressions above
$k$ is a parameter such that 
the order of the splitting formula that we use  for Hamiltonian simulation is $2k+1$,
$c(k)$ increases with $k$,
and $\eta>0$ is arbitrary. In fact, one can optimize the expressions above with 
respect to $k$, as in \cite{PZ12}. We do not pursue this direction, since it would overly 
complicate our analysis and the details of the quantum algorithm. In general, choosing $k=2$ is sufficient.

Furthermore, a direct consequence of our algorithm is that the state it produced at the end of the $L$th stage
{\it overlaps}\footnote{We use the expression \lq\lq overlaps within $1-\delta$\rq\rq to denote that
the square of the magnitude of the projection of one state onto the other is 
bounded from below by $1 - \delta$, for $\delta\in (0,1)$.}
with the ground state of the discretized Hamiltonian within $1-O\left((Cd)^{-\frac{2-\eta}{1-\eta}}\right)$. We modify
the algorithm to prepare approximations of the ground state of the discretized Hamiltonian that overlap
within $1-O(\delta)$, for $\delta = o\left((Cd)^{-\frac{2-\eta}{1-\eta}}\right)$. The resulting cost is
\[
 c(k)\cdot C^{1+\frac{1}{2k}} \cdot d^{2-\frac{1}{2k}} \cdot \e^{-\left(3+\frac{1}{2k} \right)}
\cdot \delta^{ - 1 - \frac{1}{2k} - \frac{1}{2-\eta} - \frac{1}{k(2-\eta)}}
\]
and the number of qubits is
\[
3 \log \e^{-1} +  \log \delta^{-1} + \Theta \left(d \log\e^{-1}\right).
\]

\section{Discretization error}

The finite difference method is frequently used to discretize 
partial differential equations, and approximate their solutions. 
The method with mesh size $h = \frac{1}{n+1}$ yields an 
$n^d \times n^d$ matrix $M_h = -\frac{1}{2}\Delta_h + V_h$, where $\Delta_h$ denotes 
the discretized Laplacian and $V_h$ the diagonal matrix whose  
entries are the evaluations of the potential on a regular grid 
with mesh size $h = \frac{1}{n+1}$.

$M_h$ is a symmetric positive definite and sparse matrix. For a potential
function $V$ that has uniformly bounded first order partial derivatives, we have 
\cite{Weinberger56,Weinberger58}
\begin{equation}
\label{eq:cont-discr}
|E_0 - E_{h,0}| \leq c_1 d h,
\end{equation}
where $E_{h,0}$ is the smallest eigenvalue of $M_h$.
Consider $\hat E_{h,0}$ such that 
\begin{equation}
\label{eq:app1}
|E_{h,0} -\hat E_{h,0}| \leq c_2 d h.
\end{equation}
Then we have $|1 - \frac{\hat E_{h,0}}{E_0}| \leq c^\prime h$,
where $c^\prime$ is a constant. The 
inequality follows by observing that 
$E_0 \geq d\pi^2/2$, for any $V\geq 0$. Hence we estimate $E_0$ with 
\textit{relative error} $O(\e)$  by taking $h\le\e$ and approximating the lowest eigenvalue 
of $M_h$ with absolute error $O(dh)$.

\section{Quantum Algorithm}

We consider the Hamiltonian $H_\ell = -\frac{1}{2}\Delta + \frac{\ell V}{L}$,
and the respective discretized Hamiltonian $M_{h,\ell} = - \frac{1}{2} \Delta_h + \frac{\ell V_h}{L}$,
where $\ell=1,2,\ldots ,L$ and the value of $L$ will be chosen appropriately later. We proceed in
$L$ stages. In the $\ell$th stage, we solve the eigenvalue problem for $H_\ell$ 
(and $M_{h,\ell}$)
and pass the results to the next stage. In each stage the  eigenvalue problem is solved
using phase estimation.

In the following section we present some of the properties of phase estimation that we need for our algorithm.
We present our algorithm in sections \ref{sec:eiv} and \ref{sec:ground_state}. The former section deals with the 
estimation of the ground state energy of $H$ and the latter section deals with the estimation of the 
ground state eigenvector of $M_h$.

\subsection{Phase estimation improves approximate eigenvectors}
\label{sec:phase_estimation}

Phase estimation \cite[Fig. 5.2, 5.3]{NC00} is used to approximate the eigenvalues of unitary matrices provided certain conditions hold. The input is the eigenvector that corresponds
to the eigenvalue of interest and the eigenvalue estimate is computed using a measurement outcome at the end of the algorithm. Approximate eigenvectors can also be used input as as long as the magnitude of their projection on the true eigenvector is not exponentially small \cite{AL99}. In such a case, after the measurement, the quantum register that was holding the approximate eigenvector now holds a new state that is also 
an approximate eigenvector, often an improved one. This was observed in \cite{TM01} without, however, showing rigorous error estimates and conditions. In this section we study
the eigenvector approximation using phase estimation.

Let $A$, $\| A\|\le R$, be an $n^d\times n^d$  Hermitian matrix. Then the eigenvalues of $U=e^{iA/R}$ have the form
$e^{i\lambda /R}$, where $\lambda$ denotes an eigenvalue of $A$. Equivalently $e^{i\lambda /R}=e^{2\pi i \phi_\lambda}$, where $\phi_\lambda = \lambda /(2\pi R)\in [0,1)$ is the phase 
corresponding to $\lambda$. 

Besides the (approximate) eigenvector, phase estimation uses  matrix exponentials of the form  $U^\tau=e^{iA\tau /R}$ to accomplish its task. Frequently, approximations $\tilde U_\tau$
are used instead. For instance, when $A$ is given as a sum of Hamiltonians each of which can be implemented efficiently one can use a splitting formula \cite{Suzuki90,Suzuki92}
to approximate $U_\tau$. Let the initial state and the matrix exponentials in phase estimation be as follows:
\begin{itemize}
\item \textit{Initial state:} We have $\ket{0}^{\otimes b}$ in the top register,
that deals with the accuracy, and $ \ket{\psi_{\textrm{in}}}$ in the bottom register.
\item \textit{Matrix exponentials:} We have a unitary matrix $\tilde U_{2^t}$
approximating
$U^{2^{t}}= e^{iA2^t/R}$, for $t=0,1, \ldots ,b-1$. 
Assume that the total error in the approximation of the exponentials is 
bounded by $\e_H$, i.e.
\begin{equation}\label{eq:expErr}
\sum_{j=0}^{b-1} \norm{U^{2^j} - \tilde U_{2^j}} \leq \e_H ,
\end{equation}
which implies that
\begin{equation*}
\norm{U^t - \tilde{U}_t} \leq \e_H , \,\, \textrm{for } t=0,1,\ldots ,2^b-1
\end{equation*}
\end{itemize}

Denoting  by $\{\la_j, \ket{u_j}\}_{j=0,1,\ldots ,n^d-1}$ the
eigenpairs of $A$ we have 
\begin{equation}\label{cj}
\ket{\psi_{\textrm{in}}} = \sum_{j}c_j \ket{u_j}.
\end{equation}

Given a relatively rough approximation of the eigenvector of interest as input in the bottom register of phase estimation,
we show conditions under which the bottom register at the end of the algorithm,
and after the measurement, holds a state that is an improved approximation of the eigenvector of interest. 
To simplify matters we proceed in two steps. If the top register is $b$ qubits long as we discussed above, then the conditions 
for the improvement are shown in Proposition \ref{thm1} in the Appendix, but the resulting success probability is not satisfactory.
To increase the success probability we have extended the top register by $t_0$ qubits and modified part of the proof of Proposition \ref{thm1} 
to obtain the theorem below. (Proposition \ref{thm1} is really a simplified version of this theorem.)


\begin{theorem} 
\label{thm1mod}
Let $\ket{\psi_{m^\prime}}$ be the final state in the bottom register after measuring $m^\prime$ on the top register of the phase estimation 
with initial state $\ket{0}^{\otimes (b+t_0)} \ket{\psi_{\rm in}}$ and unitaries $\tilde{U}_t$, $t=0,1,\ldots ,2^{b+t_0}-1$ and $t_0 \geq 1$. 
Let $c_0=\inner{u_0}{\psi_{\rm in}}$ and $c_0^\prime=\inner{u_0}{\psi_{m^\prime}}$, where $\ket{u_0}$ is the ground state eigenvector.
If
\begin{itemize}
\item $b$ is such that the phases satisfy $\left|\phi_j - \phi_0\right| > \frac{5}{2^{b}}$ 
for all $j=1,2,\ldots ,n^d-1$,
\item $|c_0|^2 \geq \frac{\pi^2}{16}$,
\end{itemize} 
then with probability $p \geq |c_0|^2 \left(1 - \frac{1}{2(2^{t_0} - 1)}\right) - \left(\frac{5\pi^2}{2^5} + 
\frac{1-\frac{\pi^2}{16}}{2^5}\right) \cdot \frac{1}{2^{t_0}} - 2\sum_{j = 0}^{b+t_0 -1} 
\norm{U^{2^j} - \tilde U_{2^j}}$ we obtain the measurement outcome $m^{\prime}$ satisfying
\begin{itemize}
\item $m^\prime \in \mathcal{G}:= \left\{m\in\{0,1,\ldots 2^{b+t_0}-1\} \,\, : \,\, \left|\phi_0 - 
\frac{m^\prime}{2^{b+t_0}}\right| \leq \frac{1}{2^b}\right\}$
\end{itemize} 
and
\begin{itemize}
\item if $1-|c_0|^2 \leq \gamma \e_H$ then $1 - |c_0^{\prime}|^2 \leq 
(\gamma + 14 )\e_H$
\item if $1-|c_0|^2 \geq \gamma \e_H^{1-\eta}$, for $\eta \in (0,1)$, 
then $|c_0^{\prime}| \geq |c_0|$
\end{itemize}
where $\gamma>0$ is a constant. 
\end{theorem}

\begin{proof} 
Just like in Proposition \ref{thm1} we derive an equation similar to  (\ref{eq:1}), namely, 
\[|c_0^\prime|^2=
| \inner{m^\prime,\psi_{m^\prime}}{m^\prime,u_0} |^2 = \frac{\left| c_0 \alpha (m^\prime,\phi_0) + \frac{1}{2^b} \sum_{j=0}^{n^d-1} c_j \sum_{k=0}^{2^{b+t_0}-1} e^{-2\pi i m^\prime k/2^{b+t_0}}  \inner{u_0}{x_{j,k}} \right|^2}{\norm{\ket{\psi_{1,m^\prime}} + \ket{\psi_{2,m^\prime}}}^2},
\]
Without accounting for the error due to the approximations $\tilde U_t$, $ t=0,\dots,2^{b+t_0}-1$, 
with probability at least $|c_0|^2 \geq |c_0|^2 \cdot \left(1-\frac{1}{2(2^{t_0} - 1)}\right)$
we get a result $m^\prime$ such that 
$m^\prime \in \mathcal{G}$, with $\mathcal{G} = \left\{m\in\{0,1,\ldots 2^{b+t_0}-1\} \,\, : \,\, \left|\phi_0 - 
\frac{m^\prime}{2^{b+t_0}}\right| \leq \frac{2^{t_0}}{2^{b+t_0}} = \frac{1}{2^b}\right\}$, see
\cite[Thm. 11]{BHMT02}.
Moreover, according to Lemma \ref{lem:3} the probability of getting a result
$m^\prime \in \mathcal{G}$ with $\frac{|\alpha (m^\prime , \phi_j)|^2}{|\alpha (m^\prime , \phi_0)|^2}
\leq \frac{\pi^2}{32}$ for all $j\geq 1$ is at least 
$|c_0|^2 \left(1-\frac{1}{2(2^{t_0}-1)}\right) - \left(\frac{5\pi^2}{2^5} + 
\frac{1-\frac{\pi^2}{16}}{2^5}\right) \cdot \frac{1}{2^{t_0}}$.

Accounting now for the error due to the approximation of the matrix exponentials, the probability of getting an outcome $m^\prime$
that belong to $\mathcal{G}$ and also $\frac{|\alpha (m^\prime , \phi_j)|^2}{|\alpha (m^\prime , \phi_0)|^2}
\leq \frac{\pi^2}{32}$ for all $j\geq 1$ is
at least $|c_0|^2 \left(1 - \frac{1}{2(2^{t_0} - 1)}\right) - \left(\frac{5\pi^2}{2^5} + 
\frac{1-\frac{\pi^2}{16}}{2^5}\right) \cdot \frac{1}{2^{t_0}} - 2\sum_{j = 0}^{b+t_0 -1} 
\norm{U^{2^j} - \tilde U_{2^j}}$, see \cite[pg. 195]{NC00}. From now on we consider only such 
outcomes. 

As in Proposition \ref{thm1}, we have the equivalent of (\ref{eq:3}).  Using 
$\frac{|\al(m^{\prime},\phi_j)|^2}{|\al(m^{\prime},\phi_0)|^2} \leq \frac{\pi^2}{32}$, $j \geq 1$,  we obtain
\begin{eqnarray}
|c_0^{\prime}| &>& |c_0| \left(\frac{1}{\sqrt{|c_0|^2 + \frac{\pi^2}{32}\cdot \sum_{j=1}^{n^d-1} |c_j|^2 }}- 7 \frac{\e_H}{|c_0|} \right) \nonumber \\
 &=& |c_0| \left(\frac{1}{\sqrt{|c_0|^2 + \frac{\pi^2}{32}\cdot (1 - |c_0|^2) }}- 7 \frac{\e_H}{|c_0|} \right) ,
 \label{eq:4mod}
\end{eqnarray}

since $\sum_{j=0}^{n^d-1}|c_j|^2 = 1$.

Now we examine the different cases, depending on the
magnitude of $|c_0|$.

\textit{Case 1:} $1-|c_0|^2 \leq \gamma \e_H$, for a constant $\gamma$. 
Then (\ref{eq:4mod}) becomes
\[
|c_0^{\prime}| > |c_0| \left(1 - 7 \frac{\e_H}{\sqrt{1-\gamma \e_H}} \right) ,
\] 
because $f(x) = \frac{1}{\sqrt{x + \frac{\pi^2}{32} (1 - x)}}$ is a monotonically decreasing function
and $|c_0|^2 \leq 1$. Hence
\begin{eqnarray*}
|c_0^{\prime}|^2 &>& |c_0|^2 \left(1 - \frac{14}{\sqrt{1-\gamma \e_H}} \e_H + \frac{49}{1-\gamma \e_H} \e_H^2\right) \\
&\geq& (1-\gamma \e_H) \cdot \left(1 - \frac{14}{\sqrt{1-\gamma \e_H}} \e_H + \frac{49}{1-\gamma \e_H} \e_H^2\right) \\
&=& 1 - \gamma \e_H - 14 \e_H \sqrt{1-\gamma \e_H} + 49 \e_H^2 \\
&\geq & 1 - \gamma \e_H - 14 \e_H + 49 \e_H^2 \geq 1 - (\gamma + 14) \e_H ,
\end{eqnarray*}
since $1- \gamma \e_H < 1$. This concludes the first part of the theorem.

\textit{Case 2:} $1-|c_0|^2 \geq \gamma \e_H^{1-\eta}$, for some $\eta \in (0,1)$ and $\gamma >0$.
Then (\ref{eq:4mod}) becomes
\[
|c_0^{\prime}| > |c_0| \left(\frac{1}{\sqrt{1- \left(1 - \frac{\pi^2}{32}\right) \gamma \e_H^{1-\eta}}} - 7 \frac{\e_H}{\pi/ 4} \right),
\] 
because $f(x) = \frac{1}{\sqrt{x + \frac{\pi^2}{32} (1 - x)}}$ is a
monotonically decreasing function and $|c_0|^2 \geq \frac{\pi^2}{16}$.

Note that $\frac{1}{\sqrt{1-a}} \geq \sqrt{ 1 + a }$, for $ |a| \leq 1$. Hence
\begin{eqnarray*}
|c_0^{\prime}|^2 &>& |c_0|^2 \left( 1 + \left(1 - \frac{\pi^2}{32}\right)\gamma \e_H^{1-\eta} - \frac{56}{\pi} \e_H \sqrt{1 + \left( 1 - \frac{\pi^2}{32}\right)\gamma \e_H^{1-\eta}} + \frac{28^2}{\pi^2} \e_H^2\right) \\
&>& |c_0|^2 \left(1 + \left(1 - \frac{\pi^2}{32}\right)\gamma \e_H^{1-\eta} - O(\e_H)\right)> |c_0|^2.
\end{eqnarray*}
This concludes the proof, since we can discard the $O(\e_H)$ terms for $\e_H$ sufficiently small.
\end{proof}

\subsection{Approximation of the ground state energy}
\label{sec:eiv}
 
As we already indicated our algorithm goes through $L$ stages; recall Fig.~\ref{Fig:Alg_simple}. 
In the $\ell$th state the discretized potential is $\ell \cdot V_h/L$ and we
consider the Hamiltonian $M_{h,\ell} = -\frac{1}{2}\Delta_h + \frac{\ell}{L} V_h$. Let 
$\ket{u_{0,\ell}}$ be its ground state eigenvector.
We approximate the ground state energy (the minimum eigenvalue) of $M_{h,\ell}$
within relative error $\e$, $\ell = 1,2, \ldots ,L$. 
We show how to set up the parameters of the algorithm so that
the state produced after the measurement
in the bottom register at the end of stage $\ell - 1$ is an approximation of both $\ket{u_{0,\ell-1}}$
and $\ket{u_{0,\ell}}$ and, therefore, can be used as input in the $\ell$th stage. We repeat this procedure
until $\ell=L$. The purpose of the stages $\ell=1,\dots, L-1$ is to gradually produce a relatively good
approximation of the ground state eigenvector $\ket{u_{0,L}}$ of the Hamiltonian $M_h=M_{h,L}$, 
with high probability. The last stage computes the approximation of the ground state energy of $M_h$.



We first introduce some useful notation.
Phase estimation requires two quantum registers \cite[Fig.~5.2,~5.3]{NC00}. 
The top register determines the accuracy and the probability of
success of the algorithm and
the bottom register holds an approximation of the
ground state of $M_{h,\ell}$. Let $\ket{\psi_{{\rm in},\ell}}$ be
the state on the bottom register at 
the beginning of the $\ell$th stage and let $\ket{\psi_{{\rm out},\ell}}$ be
the state on the same register at the end of the $\ell$th stage; see Fig.~\ref{fig:alg-diagram}.

\begin{figure}[H]
\centering
\subfigure[First stage of the algorithm]{
\Qcircuit @R=1em @C=1em {
& & & & & & & & & \\
& & & & & & & & & \\
& & & & & & \qw & \multigate{5}{ \textrm{\parbox[t]{0.12\linewidth}{Phase Estimation on $W_{h,1}$}}} & \meter
\gategroup{6}{10}{8}{10}{.7em}{\}} \gategroup{6}{6}{8}{6}{.7em}{\{} \gategroup{3}{6}{5}{6}{.7em}{\{} & & \\
& & \ket 0^{\otimes (b+t_0)} & & & &  \vdots & & \vdots & & & \\
& & & & & & \qw & \ghost{ \textrm{\parbox[t]{0.12\linewidth}{Phase Estimation on $W_{h,1}$}}} & \meter & &  \\
& & & & & & \qw & \ghost{ \textrm{\parbox[t]{0.12\linewidth}{Phase Estimation on $W_{h,1}$}}} & \qw &  \qw & & \\
\ket{\psi_{\rm in, 1}} \equiv \ket{u_{0,0}} & & & & & &\vdots & & \vdots & & &  & \ket{\psi_{\rm out, 1}} \\
& & & & & & \qw & \ghost{ \textrm{\parbox[t]{0.12\linewidth}{Phase Estimation on $W_{h,1}$}}} & \qw & \qw & & \\
& & & & & & & & & & \\
&  \\
&
}
}
\\ 
\subfigure[Diagram of the algorithm for stages $\ell-1$  and $\ell$, with $\ell = 2,3, \ldots L$. 
The phase estimation on stage $\ell$ runs for the unitary $W_{h,\ell} = e^{-iM_{h,\ell}/R}$, with $M_{h,\ell} = -\frac{1}{2}
\Delta_h + \ell \cdot \frac{V_h}{L}$]{
\Qcircuit @R=1em @C=1em {
& & & & & \rm Stage \,\, \ell -1 \gategroup{2}{4}{2}{8}{.7em}{^\}} \gategroup{6}{8}{8}{8}{.7em}{\}} & 
& & & & & & & & & & &  & &  \rm Stage \,\, \ell \gategroup{2}{18}{2}{22}{.7em}{^\}} \gategroup{6}{4}{8}{4}{.7em}{\{} \gategroup{3}{4}{5}{4}{.7em}{\{} \gategroup{3}{18}{5}{18}{.7em}{\{} \gategroup{6}{18}{8}{18}{.7em}{\{} 
\gategroup{6}{22}{8}{22}{.7em}{\}} \\
\ar@{-->}[rrd]|<{\rm top\,\, register} & & & & & & & & & & & \ar@{-->}[lllld]|<{\rm measurement} & & & & & & & & & & & & &\\
& & & & \qw & \multigate{5}{ \textrm{\parbox[t]{0.12\linewidth}{Phase Estimation on $W_{h,\ell-1}$}}} & \meter & & & & & & & & & & & & \qw &  \multigate{5}{ \textrm{\parbox[t]{0.12\linewidth}{Phase Estimation on $W_{h,\ell}$}}}  & \meter \\
\ket 0^{\otimes (b+t_0)} & & & &  \vdots & & \vdots & & & & & & & & \ket 0^{\otimes (b+t_0)} & & & & \vdots & & \vdots \\
& & & & \qw & \ghost{ \textrm{\parbox[t]{0.12\linewidth}{Phase Estimation on $W_{h,\ell-1}$}}} & \meter & & & & & & & & & & & & \qw & \ghost{ \textrm{\parbox[t]{0.12\linewidth}{Phase Estimation on $W_{h,\ell}$}}} & \meter  \\
& & & & \qw & \ghost{ \textrm{\parbox[t]{0.12\linewidth}{Phase Estimation on $W_{h,\ell-1}$}}} & \qw &  \qw & & & & & & & & & & & \qw  & \ghost{ \textrm{\parbox[t]{0.12\linewidth}{Phase Estimation on $W_{h,\ell}$}}} & \qw & \qw \\
\ket{\psi_{\rm in, \ell-1}}& & & &\vdots & & \vdots & & & & & &  \ket{\psi_{\rm out, \ell-1}} \equiv \ket{\psi_{\rm in, \ell}} & & & & & & \vdots & & \vdots & & & &\ket{\psi_{\rm out, \ell}} \\
& & & & \qw & \ghost{ \textrm{\parbox[t]{0.12\linewidth}{Phase Estimation on $W_{h,\ell-1}$}}} & \qw & \qw & & & & & & & & & & & \qw & \ghost{ \textrm{\parbox[t]{0.12\linewidth}{Phase Estimation on $W_{h,\ell}$}}} & \qw & \qw  \\
\ar@{-->}[rru]|<{\rm bottom\,\, register} & & & & & & & & & & & & & & & & & & & & & & & &\\
& & & & & & & & & & & & & & & & &  & & &  \\
&
}
}
\\
\subfigure[Last stage of the algorithm. The result $j$ of the measurement of the top register
provides the approximation to the ground state energy]{
\Qcircuit @R=1em @C=1em {
& & & & & & & & \\
& & & & & & & & \\
& & & & \qw & \multigate{5}{ \textrm{\parbox[t]{0.12\linewidth}{Phase Estimation on $W_{h,L}$}}} & \meter & & & & 
\ar@{-->}[lld]|<{{\rm result}\, \ket j} \gategroup{6}{4}{8}{4}{.7em}{\{} \gategroup{3}{4}{5}{4}{.7em}{\{} \gategroup{6}{8}{8}{8}{.7em}{\}}
\gategroup{3}{8}{5}{8}{.5em}{\}} \\
\ket 0^{\otimes (b+t_0)} & & & &  \vdots & & \vdots & & &\\
& & & & \qw & \ghost{ \textrm{\parbox[t]{0.12\linewidth}{Phase Estimation on $W_{h,L}$}}} & \meter & &  \\
& & & & \qw & \ghost{ \textrm{\parbox[t]{0.12\linewidth}{Phase Estimation on $W_{h,L}$}}} & \qw &  \qw & & \\
\ket{\psi_{\rm in, L}} & & & &\vdots & & \vdots & & &  & \ket{\psi_{\rm out, L}} \\
& & & & \qw & \ghost{ \textrm{\parbox[t]{0.12\linewidth}{Phase Estimation on $W_{h,L}$}}} & \qw & \qw & & \\
& & & & & & & & \\
&  \\
&
}
}
\caption{A detailed diagram describing the Repeated Phase Estimation algorithm}
\label{fig:alg-diagram}
\end{figure}
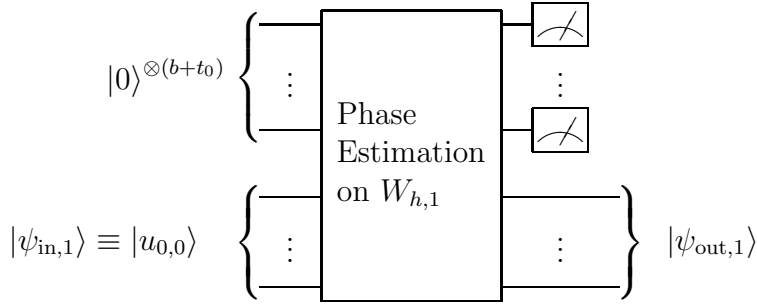
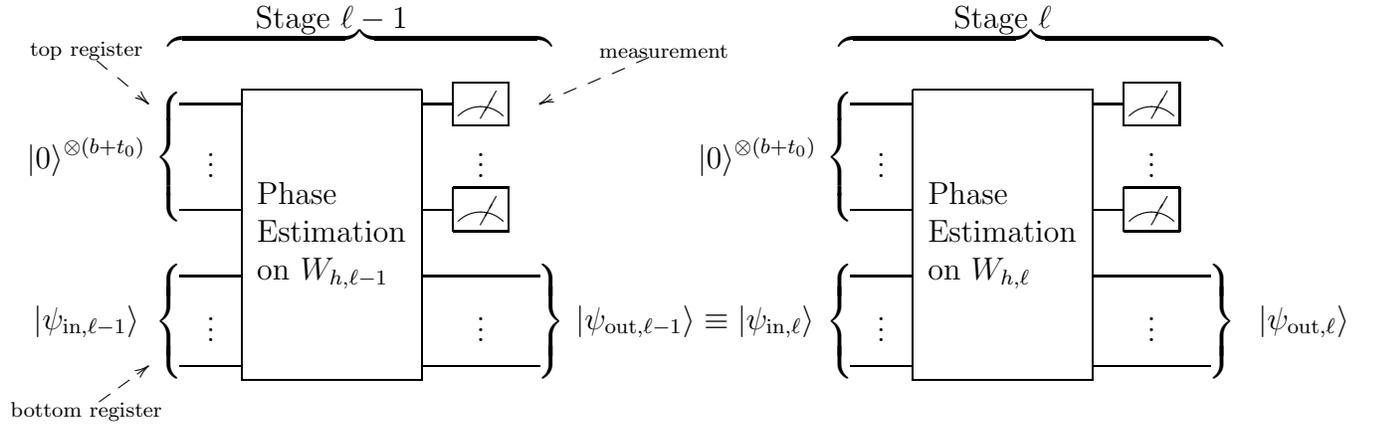
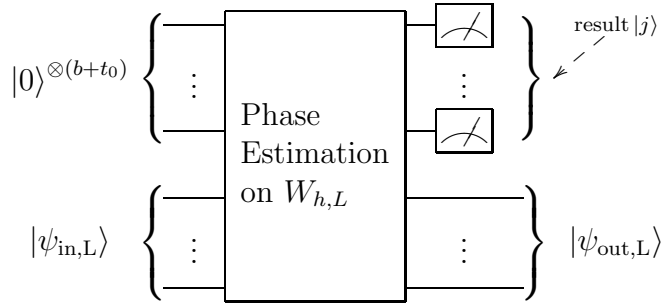

At the very beginning of the algorithm the bottom register is initialled to the state
$\ket{\psi_{{\rm in},1}} = \ket{u_{0,0}}$, 
i.e. the ground state eigenvector of the  discretized Laplacian. 
By choosing an appropriately large $L$, 
and using lower bounds for the gap between the 
first and the second eigenvalues of
Hamiltonians involving convex potentials \cite{AC11}, we ensure that
the initial state of the algorithm has a
good overlap with the ground state of $M_{h,1}$.
Theorem \ref{thm1mod} shows that  we can  maintain this good overlap between
approximate and actual ground state eigenvectors
throughout all the stages with high probability.
We use $b + t_0$ qubits on the top register. The $b$
qubits are used to control the accuracy in the eigenvalue estimates
and the $t_0$ qubits are used to boost the probability of
success of each stage.

\vskip 1pc

We provide an overview of the algorithm.

\begin{enumerate}
\item \textit{Number of qubits:} We have two resisters, the top and the bottom. The top register has $b+t_0$ qubits, while
the bottom register has $d \log_2 h^{-1}$ qubits.
\item \textit{Initial state:} The upper register
is initialized to $\ket{0}^{\otimes (b+t_0)}$.
The lower register is initialized to 
$\ket{\psi_{{\rm in},1}} = \ket{u_{0,0}}$.
\item \textit{Phase estimation:} Run phase estimation for each of the unitary matrices $W_{h,\ell} = e^{-iM_{h,\ell} /R}$ in sequence,  for
$\ell = 1,2,\ldots , L$. where
$R$ a parameter to be defined later in this section. 
In each run the top register is set to $\ket{0}^{\otimes (b+t_0)}$,
while the bottom register holds the approximate eigenstate produced in the previous stage, i.e., 
$\ket{\psi_{{\rm in},\ell}} := \ket{\psi_{{\rm out},\ell-1}}$
for $\ell = 2,3,\ldots ,L$.
\begin{itemize}
\item \textit{Implementation of exponentials:}
Implement each exponential $W_{h,\ell}^{2^j}$
with error $\e_{j,\ell}^S := \norm{W_{h,\ell}^{2^j} - \widetilde{W_{h,\ell}^{2^j}}}$, for 
$j=0,1,\ldots ,t_0+b-1$ using Suzuki's splitting 
formulas \cite{Suzuki90,Suzuki92}.
\end{itemize}
\item \textit{Output:} Let
$j\in \{0,1, \ldots ,2^{t_0 + b}-1\}$
the result of the measurement on the upper register
after the last stage.
Output $\hat E_{h,0} = 2\pi \cdot R \cdot j\cdot 2^{-(b+t_0)}$
\end{enumerate}

Let $\la_{j,\ell}$ be the $j$th eigenvalue of $M_{h,\ell}$. The phase
corresponding to this eigenvalue is 
\[
\phi_{j,\ell} = \frac{\la_{j,\ell}}{2\pi R}.
\]
Set $R = 3dh^{-2} \gg 2dh^{-2} + C \geq \norm{-\frac{1}{2} \Delta_h + V_h}$. \footnote{
Recall that
$C=o(d^5)$ as stated in the Introduction and $h = o(d^{-2})$ as stated in the next section}
This choice of $R$ guarantees that $\phi_{j,\ell} \in [0,1)$
for all $j = 0,1,\ldots ,n^d -1$ and $\ell = 1,2,\ldots ,L$.

\subsubsection{Error analysis}
\label{sec1:Error}

We know (eq.~(\ref{eq:cont-discr}) and (\ref{eq:app1})) that
we can achieve relative error $O(h)$ if we approximate
the ground energy of $M_{h,L}$ with error at most $dh$. This implies
that the algorithm has to approximate the eigenvalues $\la_{0,\ell}$
within error $dh$, for all $\ell=0,1,\ldots ,L$,
which in turn requires $\phi_{0,L}$ to be approximated 
with error $\frac{dh}{2\pi R}$. This translates
to 
$
2^{-b} \leq \frac{dh}{2\pi R},
$
which in turn leads to
\begin{equation}
\label{eq:b}
b = \left\lceil \log \frac{2R\pi}{dh} \right \rceil = \left\lceil \log (6\pi h^{-3}) \right\rceil = \log \Theta \left(h^{-3}\right).
\end{equation}

\subsubsection{Preliminary Analysis}
\label{sec:prelim}

In Theorem \ref{thm1mod} we have shown conditions 
relating the state provided as input and the state produced after the 
measurement in the bottom register of phase estimation, without
assuming a particular Hamiltonian form, as we do in this section.
In the successive application of phase estimation we intend to use the results of one stage as input to the next.
Thus we need to quantify how the results of one stage affect the success probability of the next stage
in the case of the Schr\"odinger equation with a convex potential.

The fundamental gap for Hamiltonians of the form
$-\frac{1}{2}\Delta + V$, where $V$ is a convex potential, is at least
$\frac{3\pi^2}{2d}$
\cite{AC11}. The gap between
the first and second eigenvalues of $M_{h,\ell}$, for 
$\ell = 1,2, \ldots , L$,
is reduced by $O(dh)$ \cite{Weinberger56,Weinberger58}.
Taking $h = o(d^{-2})$, the gap is at least $\frac{3\pi^2}{2d} - o(d^{-1}) \geq 
\frac{\pi^2}{d}$. \footnote{ $|\la_1 (h) - \la_0 (h)| \geq |\la_1 (h) - \la_0| \geq |\la_1 - \la_0| - |\la_1 (h) - \la_1|$
since $\la_0 (h) \leq \la_0$, where the subscript zero denotes the smallest eigenvalue of a Hamiltonian
and the subscript one denotes the second smallest eigenvalue.} 
As a result, the gap between the phases corresponding to
the first two eigenvalues is at least $\frac{\pi^2}{d\cdot 2\pi R}$. 
Setting $h < \frac{2\pi^2}{5}\cdot \frac{1}{d^2}$ , we have that
$2^{-b} < \frac{1}{5} \cdot \frac{\pi^2}{d\cdot 2\pi R}$, 
according to (\ref{eq:b}).
This leads to
$|\phi_0-\phi_j|\geq \frac{5}{2^b}$, for all $j \geq 1$.
Hence, for $h = o(d^{-2})$ the assumptions of Theorem \ref{thm1mod} hold.

Let $L = \omega (d)$ that will be specified later.
Consider the $(\ell - 1)$th stage, with initial state $\ket{\psi_{\rm in, \ell -1}}$
and Hamiltonian $M_{h,\ell-1}$. Assume
$|\inner{\psi_{\rm out,\ell-1}}{u_{0,\ell-1}}| = 1-\kappa \delta$,
where $\kappa > 0$ a constant and $\delta \in [0,1)$ a quantity satisfying
$\delta = \omega ((Cd)^2/L^2)$. That is, $\ket{\psi_{\rm in, \ell -1}}$
is not such a good approximation of the ground state eigenvector $\ket{u_{0,\ell}}$. 
Also assume that the error (\ref{eq:expErr})  due to the approximation of the matrix exponentials is $\e_H=o(\delta)$.
Then the magnitude of the projection of the resulting state  
$\ket{\psi_{\rm out,\ell-1}}$ 
of this stage  onto the
ground state eigenvector follows from Theorem \ref{thm1mod} as we show in the
corollary below.

\begin{cor} 
\label{cor:REP1step}
Let $|\inner{\psi_{\rm in,\ell-1}}{u_{0,\ell-1}}|^2 = 1- \ka \delta$,
where $\ka > 0$ is a constant, $\delta \in [0,1)$ and $\delta = \omega (\e_H)$.
Then 
\[
|\inner{\psi_{\rm out,\ell-1}}{u_{0,\ell-1}}|^2 \geq 1 - \frac{\pi^2 + 1}{32} \kappa \delta, \quad \ell\ge2 .
\]
\end{cor}

\begin{proof} We reconsider the case 2 of Theorem \ref{thm1mod}. Retracing the steps, 
we reach to
\begin{eqnarray*}
|\inner{\psi_{\rm out,\ell-1}}{u_{0,\ell-1}}|^2 &>& |\inner{\psi_{\rm in,\ell-1}}{u_{0,\ell-1}}|^2
\left( 1 + \left( 1 - \frac{\pi^2}{32}\right) \kappa  \delta - O(\e_H)\right) \\
 &= & (1 - \kappa \delta ) \left( 1 + \left( 1 - \frac{\pi^2}{32}\right) \kappa  \delta - O(\e_H)\right) \\ 
&= & 1 - \frac{\pi^2}{32} \kappa \delta - O(\e_H) + \kappa \delta \cdot O(\e_H) \\
& \geq & 1 - \frac{\pi^2 + 1}{32} \kappa \delta ,
\end{eqnarray*}
where the last inequality is due to  $\delta = \omega (\e_H)$.
\end{proof}

Observe after stage $\ell-1$ is complete, that phase estimation has improved the approximation of $\ket{u_{0,\ell-1}}$.
Note that $\ket{\psi_{{\rm out}, \ell-1}} = \ket{\psi_{{\rm in}, \ell}}$  and $|\inner{u_{0,\ell}}{\psi_{{\rm in}, \ell}}|$ determines the success probability of the $\ell$th stage.
To calculate this, we need to take into account the projection of
the ground state eigenvector $\ket{u_{0,\ell-1}}$ onto $\ket{u_{0,\ell}}$.

Taking into account the lower bound on the gap between the first 
two eigenvalues of $M_{h,\ell}$, we express $\ket{u_{0,\ell-1}}$
in terms of the eigenstates of $M_{h,\ell}$
to get
\begin{equation}
\label{eq:gap}
\norm{V/L}^2_\infty \geq 
(1-|\inner{u_{0,\ell-1}}{u_{0,\ell}}|^2) \left(\frac{\pi^2}{d}\right)^2 
\Rightarrow |\inner{u_{0,\ell-1}}{u_{0,\ell}}|^2 \geq  1 - \left(\frac{Cd}{\pi^2 L}\right)^2,
\end{equation}
for $\ell = 2,\ldots ,L$, see \cite{Pap07}.


\begin{lem} 
\label{lem:RPE2step}
Let 
\begin{eqnarray*}
|\inner{\psi_{\rm out,\ell-1}}{u_{0,\ell-1}}|^2 & \geq &  1 - \kappa^{\prime} \delta \\
|\inner{u_{0,\ell-1}}{u_{0,\ell}}|^2 & \geq & 1 - \left(\frac{Cd}{\pi^2 L}\right)^2,
\end{eqnarray*}
with 
$\delta = \omega \left(\left(\frac{Cd}{\pi^2 L}\right)^2\right)$. Then
\[
|\inner{\psi_{\rm out,\ell-1}}{u_{0,\ell}}|^2 \geq 1 - \kappa^{\prime} \delta - o(\delta), \quad \ell\ge 2,
\]
where $\kappa^\prime>0$  is a constant.
\end{lem}

\begin{proof}
Let $\theta_1 := \arccos |\inner{\psi_{\rm out,\ell-1}}{u_{0,\ell-1}}|$ and $\theta_2 = \arccos |\inner{u_{0,\ell-1}}{u_{0,\ell}}|$.

\vspace{5mm}
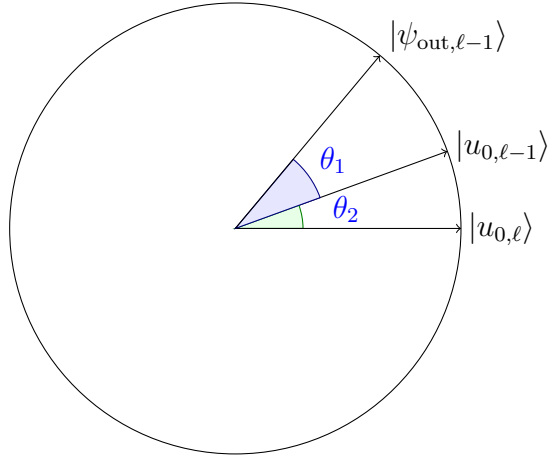
\begin{figure}[H]
\centering
\begin{tikzpicture}[scale=3]
\draw (0,0) circle (1cm);
\draw [->] (0,0) -- (0.939, 0.342);
\filldraw [fill = green!10!white, draw =green!50!black] (0,0) -- (3mm,0mm) arc (0:20:3mm) -- cycle;
\draw [->](0,0) -- (1,0);
\filldraw [fill = blue!10!white, draw =blue!50!black] (0,0) -- (3.758mm,1.368mm) arc (20:50:4mm) -- cycle;
\draw [->] (0,0) -- (0.64,0.766);
\draw[very thick,blue] (10:.5cm)  node {$\theta_2$} ;
\draw[very thick,blue] (35:.53cm)  node {$\theta_1$} ;
\draw[very thick,black] (17:1.24cm)  node {$\ket{u_{0,\ell-1}}$} ;
\draw[very thick,black] (0:1.18cm)  node {$\ket{u_{0,\ell}}$} ;
\draw[very thick,black] (42:1.27cm)  node {$\ket{\psi_{\rm out,\ell-1}}$} ;
\end{tikzpicture}
\caption{The magnitude of the projection of $\ket{\psi_{\rm out, \ell -1}}$ onto $\ket{u_{0,\ell}}$
in the worst case}
\label{Fig:vectors}
\end{figure}
\vspace{5mm}

Then $|\inner{\psi_{\rm out,\ell-1}}{u_{0,\ell}}|^2 \geq \cos^2 (\theta_1 + \theta_2)$ (see Fig. \ref{Fig:vectors}).
Note that
$$\cos^2 (\theta_1 + \theta_2) = \frac{1}{2} [1+ \cos(2 (\theta_1 + \theta_2))] =
\frac{1}{2} [1+ \cos (2\theta_1) \cos (2\theta_2) - \sin (2\theta_1)\sin (2\theta_2)].$$
Now $\cos^2 (\theta_1) = \frac{1}{2} \left[ 1 + \cos (2 \theta_1) \right] \geq 1 - \kappa^{\prime} \delta$, which leads to 
\[\cos (2\theta_1) \geq 1 - 2 \kappa^{\prime}\delta.\]
Similarly 
\[\cos (2\theta_2) \geq 1 - 2 \left(\frac{Cd}{\pi^2 L}\right)^2.\]
Furthermore 
$$\sin^2 (2\theta_1 ) = 1- \cos^2 (2\theta_1) \leq 1 - \left[ 1 - 2 \kappa^{\prime} \delta
\right]^2 \leq 4 \kappa^{\prime} \delta ,$$ 
and similarly $\sin^2 (2\theta_2 ) \leq 4 \left(\frac{Cd}{\pi^2 L}\right)^2$.
According to the above 
\begin{eqnarray*}
\cos^2 (\theta_1 + \theta_2) & \geq &\frac{1}{2} \left[ 1 + \left(1 - 2 \kappa^{\prime} \delta \right) \cdot \left(1 - 2 \left(\frac{Cd}{\pi^2 L}\right)^2\right) - \sqrt{4 \kappa^{\prime} \delta} \cdot \sqrt{4 \left(\frac{Cd}{\pi^2 L}\right)^2}\right]  \\
&\geq  & 1- \kappa^{\prime} \delta + 2 \kappa^{\prime} \delta \left(\frac{Cd}{\pi^2 L}\right)^2
- \left(\frac{Cd}{\pi^2 L}\right)^2 - 2 \sqrt{\kappa^{\prime}} \cdot \sqrt{\delta \cdot \left( \frac{Cd}{\pi^2 L}\right)^2} \\
& \geq & 1 - \kappa^{\prime} \delta - o(\delta)
\end{eqnarray*}
since $\left(\frac{Cd}{\pi^2 L}\right)^2 = o(\delta)$ and $\sqrt{\delta \left(\frac{Cd}{\pi^2 L}\right)^2} = \sqrt{\delta \cdot o(\delta)} = o(\delta)$.
\end{proof}

Consider errors $\e^S_{j,\ell}$ for the exponentials $W_{h,\ell}^{2^j}$
such that the total error for each stage is 
$\sum_{j=0}^{t_0 + b-1} \e^S_{j,\ell} \le \e_H= O\left(\left(\frac{Cd}{\pi^2 L}\right)^2 \right)$.
We use Corollary \ref{cor:REP1step}, and Lemma \ref{lem:RPE2step} to get
\begin{equation}\label{eq:inTrue}
|\inner{\psi_{\rm out,\ell-1}}{u_{0,\ell}}|^2 = 1- \frac{\pi^2 +1}{32} \kappa \delta - o(\delta) \geq 1 - \left(\frac{\pi^2 +1}{32} + 10^{-3}\right) \kappa \delta ,
\end{equation}
where $\kappa>0$ is a constant, $\delta=\omega\left( \left( \frac{Cd}{\pi^2L} \right)^2 \right)$ and  $\ell\ge 2$.

\subsubsection{Initial state}
\label{sec:initial_state}

The initial state of our algorithm is the ground state eigenvector 
$\ket{u_{0,0}}$
of the discretized Laplacian. Hence, $\ket{\psi_{{\rm in}, 1}}:= \ket{u_{0,0}} = \ket{z}^{\otimes d}$,
where $\ket z$ is the ground state eigenvector of the $n \times n$ matrix corresponding
to the (one dimensional) discretization of the second derivative. The coordinates of $\ket z$ are 
\[
z_j = \sqrt{2h} \sin (j\pi h), \quad \rm for \,\, j = 1,2, \ldots,n.
\]
and it can be implemented using the quantum Fourier transform with a number quantum operations proportional to
$\log^2 h^{-1}$ \cite{KR01}.
Thus, the implementation of the initial state of the algorithm $\ket{u_{0,0}}$
requires a number of quantum operations proportional to $d\log^2 h^{-1}$.

According to  (\ref{eq:gap})  we have
\begin{equation}
\label{eq:c0}
|\inner{\psi_{\rm in,1}}{u_{0,1}}|^2 \geq 1 - \left(\frac{Cd}{\pi^2 L}\right)^2.
\end{equation}

\subsubsection{Success probability}
\label{sec:succ_prob1}

According to 
the analysis in Sections \ref{sec:prelim}, \ref{sec:initial_state}, 
and, particularly, observing that the output of one stage is the input to the next from (\ref{eq:inTrue},\ref{eq:c0})
we have 
\[
|\inner{\psi_{\rm in,\ell}}{u_{0,\ell}}|^2 \geq 1 - \delta,
\]
for any $\delta = \omega \left(\left(\frac{Cd}{L}\right)^2\right)$
and $\ell = 1,2,\ldots, L$, as long as the total error due to the approximation of the exponentials
at each stage is $O\left(\left(\frac{Cd}{L}\right)^2\right)$, i.e. $\e_{H} = O\left(\left(\frac{Cd}{L}\right)^2\right)$.
As a result,
\begin{equation}
\label{eq:l.b.overlap}
|\inner{\psi_{\rm in,\ell}}{u_{0,\ell}}|^2 \geq 1 - \kappa_1 \left(\frac{Cd}{L}\right)^{2-\eta},
\end{equation}
for $\kappa_1$ a constant and any $\eta > 0$.

From Theorem \ref{thm1mod},
the total probability of success of the algorithm after $L$ steps is
\begin{equation}
\begin{split}
P_{\rm total} \geq 
\left( \min_{\ell = 1,2,\ldots ,L} 
|\inner{\psi_{\rm in, \ell}}{u_{0, \ell }}|^2 \cdot  
\left(1 - \frac{1}{2(2^{t_0}-1)}\right) \right.
- \\ 
\left. \left(\frac{5\pi^2}{2^5} + \frac{1-\pi^2/16}{2^5} \right) \frac{1}{2^{t_0}} - 
 2\sum_{j=0}^{b+t_0-1}\norm{W_{h,\ell}^{2^j} - 
\widetilde{W_{h,\ell}^{2^j}}} \right)^L. 
 \label{eq:prob_first1}
\end{split}
\end{equation}
We use Suzuki's decomposition formulas to approximate the 
exponentials  $W_{h,\ell}^{2^j}$, for 
$j=0,1,\ldots ,t_0+b-1$; see \cite{Suzuki90,Suzuki92}.
We select 
\[
\e^{S}_{j,\ell} := 2^{j-(b+t_0)}\cdot \left(\frac{Cd}{L}\right)^2.
\]
The total error at each stage is
\[
\sum_{j = 0}^{b+t_0 - 1} \e^{S}_{j,\ell}  \leq \left(\frac{Cd}{L}\right)^2.
\]
As in (\ref{eq:expErr}), the choice of $\e^{S}_{j,\ell}$ also implies that
\[
\norm{W_{h,\ell}^{k}-\widetilde{W^{k}_{h,\ell}}} \leq \left(\frac{Cd}{L}\right)^2 ,
\]
for $k = 0,1,\ldots 2^{t_0+b}-1$.
Using (\ref{eq:l.b.overlap}) and the inequality above, (\ref{eq:prob_first1}) becomes
\begin{equation}
\begin{split}
\label{eq:prob_first}
P_{\rm total} \geq 
\left( \left( 1 - \kappa_1 \left( \frac{Cd}{L}\right)^{2-\eta}\right) \cdot  
\left(1 - \frac{1}{2(2^{t_0}-1)}\right) - 
\right. \\ \left. 
\left(\frac{5\pi^2}{2^5} + \frac{1-\pi^2/16}{2^5} \right) \frac{1}{2^t_0} - 
2 \left(\frac{Cd}{L}\right)^2
\right)^L. 
\end{split}
\end{equation}
Select $t_0$  to satisfy $\frac{1}{2(2^{t_0}-1)} \leq 
\left(\left(\frac{Cd}{L}\right)^{2-\eta}\right)$,
by setting
\begin{equation}
\label{eq:t0}
t_0 = \left\lceil \log \left( \frac{L}{Cd} \right)^{2-\eta}\right\rceil .
\end{equation}
Then (\ref{eq:prob_first}) becomes
\begin{eqnarray*}
P_{\rm total} &\geq & 
\left[ \left( 1 - \kappa_1 \left( \frac{Cd}{L}\right)^{2-\eta}\right) \cdot  
\left(1 - \left( \frac{Cd}{L} \right)^{2-\eta} \right) - 
\frac{6\pi^2}{2^5} \left(\frac{Cd}{L}\right)^{2-\eta} - 
2 \left(\frac{Cd}{L}\right)^2
\right]^L \\
&\geq & \left[ 1 - \left( \kappa_1 + 3 + \frac{6\pi^2}{2^5}\right) \left( \frac{Cd}{L}\right)^{2-\eta} \right]^L
\end{eqnarray*}
Let $\kappa_2 := \kappa_1 + 3 + \frac{6\pi^2}{2^5}$ for brevity. Then
\[P_{\rm total} \geq  \left[ 1 - \kappa_2 \left( \frac{Cd}{L}\right)^{2-\eta} \right]^L  \geq \left( 1 - \frac{(\kappa_2 \cdot (Cd)^{2-\eta}/ L^{1-\eta})^2}{L} \right) \cdot  e^{-\kappa_2 \cdot  (Cd)^{2-\eta} / L ^{1-\eta}},\]
since $\left(1-\frac{x}{n}\right)^n \geq \left(1 - \frac{x^2}{n}\right)\cdot e^{-x} $ for $|x|\leq n$ and $n >1$. In our case
$x = \kappa_2 \cdot \frac{(Cd)^{2-\eta}}{L^{1-\eta}}$, $n = L$ and the inequality is satisfied when $L$ is sufficiently
large, which we choose by
\begin{equation}\label{eq:L1}
L = (Cd)^{(2-\eta)/(1-\eta)}.
\end{equation}
Then the probablity of success becomes
\[
P_{\rm total} \geq \left(1 - \frac{\kappa_2^2}{L}\right) \cdot e^{-\kappa_2} = \left(1 - \frac{\kappa_2^2}{ (Cd)^{(2-\eta)/(1-\eta)}}\right) \cdot e^{-\kappa_2} \geq \frac{1}{2} \cdot e^{-\kappa_2} = \Omega (1),
\]
for $d$ sufficiently large. Hence, the choice of $L$ in (\ref{eq:L1}) guarantees that the algorithm has 
constant probability of success. Moreover, we can get $P_{\rm total} \geq 3/4$
if we repeat the algorithm and choose the median
as the final  result. 


\subsubsection{Cost}

We use Suzuki splitting formulas \cite{Suzuki90,Suzuki92} in order to 
express the exponentials
$W_{h,\ell}^{2^j} = 
e^{-iM_{h,\ell} /R}$, with respect to exponentials involving
either $-\Delta_h$ or
$V_h$. 

The Suzuki splitting formula $S_{2k}$ of order $2k+1$ for $k=1,2,3, \ldots$
approximates exponentials of the form 
$e^{-i (A+B) \Delta t}$, where $A$, $B$ are Hermitian matrices. The formula is
constructed recursively by
\begin{eqnarray*}
S_2 (A,B,\Delta t) &=& e^{-iA\Delta t/2} e^{-iB\Delta t} e^{-iA\Delta t/2} \\
S_{2k} (A,B,\Delta t) &= & [S_{2k-2} (A,B,p_k \Delta_t)]^2 \cdot S_{2k-2} (A,B,(1-4p_k) \Delta t) \\
&\cdot &  [S_{2k-2} (A,B,p_k \Delta_t)]^2,
\end{eqnarray*}
where $k = 2,3,\ldots$. Unfolding the recurence and according to \cite[Thm.~1]{PZ12} we obtain
the approximation of $W^{2^j}_{h,\ell}$ by
\[
\widetilde{W^{2^j}_{h,\ell}} = e^{-i H_{1,\ell} s_{j,0}} \cdot e^{-i H_{2,\ell} z_{j,1}} \cdot e^{-i H_{1,\ell} s_{j,1}}
\cdots  e^{-i H_{2,\ell} z_{j,K_{j,\ell}}} \cdot e^{-i H_{1,\ell} s_{j,K_{j,\ell}}},
\] 
where $s_{j,0},\ldots ,s_{j,K_{j,\ell}}$, $z_{j,1},\ldots ,z_{j,K_{j,\ell}}$
and $K_{j,\ell}$ are parameters with $j = 0,1,\ldots ,t_0 + b - 1$ 
and $\ell = 1,2, \ldots , L$.
The Hermitian matrices $H_{1,\ell}$
and $H_{2,\ell}$ are defined for each stage $\ell$ as
$H_{1,\ell} = - \Delta_h / (2R)$ and $H_{2,\ell} = \ell \cdot V_h / (LR)$.

The exponentials involving $-\Delta_h$ can be implemented in
$O(d\log^2 h^{-1})$ quantum operations using the quantum Fourier transform  \cite{KR01,Wick94}.The exponentials
involving $V_h$ can be implemented with two bit queries. Approximately half 
of the exponentials involve $-\Delta_h$ and the other half involve $V_h$. Consequently
the number of exponentials provides a good estimate on the cost of the algorithm.
The application of Suzuki's splitting formulas in the case of the time independent Schr\"odinger equation is discussed in detail in \cite{PPTZ12}.

Let $N_{j,\ell}$ be the total number of exponentials required 
for the simulation of $W_{h,\ell}^{2^j}$, $N_\ell$  be the total number of the exponentials  
required for the $\ell$th stage and $N$ be
the total number of exponentials required for all the stages of the
algorithm. Using the results of \cite{PZ12} we have
\[
N_{\ell} = \sum_{j=0}^{t_0 + b -1 }N_{j,\ell} \leq \sum_{j=0}^{t_0 + b -1} 2 \cdot 5 \cdot 5^{k-1} \cdot \norm{-\Delta_h/R}_2 \cdot 
2^{j} \left( \frac{4e\cdot 2 \cdot 2^j \norm{\frac{\ell}{L} \cdot \frac{V_h}{R}}_2}{\e^{S}_{j,\ell}}
 \right)^{1/(2k)}\cdot \frac{4e\cdot 2}{3} \left(\frac{5}{3}\right)^{k-1},
\]
where $k$ is chosen so that the order of the Suzuki splitting formula is $2k+1$, $k\ge1$.  
However $\norm{-\Delta_{h}/R}_2 \leq 1$ and $\norm{\frac{\ell}{L} \cdot \frac{V_h}{R}} \leq 
\frac{\ell}{L} \cdot \frac{C}{3dh^{-2}}$, which leads to
\begin{eqnarray*}
N_{\ell} & \leq & \sum_{j = 0}^{t_0 + b -1}
\frac{80e}{3} \cdot 5^{k-1} \cdot \left(\frac{5}{3}\right)^{k-1} \cdot 1 \cdot 2^j \cdot 
\left(\frac{4 e \cdot 2 \cdot 2^j \cdot \frac{\ell}{L} \cdot \frac{C}{3dh^{-2}}}{2^{j-(b+t_0)} \cdot (Cd)^{-2/(1-\eta)}}\right)^{1/(2k)}
\\
&\leq &  
\frac{80e}{3} \cdot 5^{k-1} \cdot \left(\frac{5}{3}\right)^{k-1} \cdot \left(\frac{8e}{3}\right)^{1/(2k)} \cdot \left(2^{b+t_0}\right)^{1/(2k)}
\cdot C^{\left(1+\frac{2}{1-\eta}\right)/(2k)} \cdot d^{\left(\frac{2}{1-\eta} - 1\right) / (2k)} \\
&\cdot & h^{1/k} \cdot \sum_{j=1}^{t_0 + b - 1} 2^j \\
& \leq & \frac{80e}{3} \cdot \left(\frac{25}{3}\right)^{k-1}\cdot \left(\frac{8e}{3}\right)^{1/(2k)} \cdot (40\pi)^{1 + \frac{1}{2k}}
\cdot C^{\frac{2-\eta}{1-\eta} + \frac{5-2\eta}{2k(1-\eta)}}\cdot d^{\frac{2-\eta}{1-\eta}+ \frac{3}{2k(1-\eta)}}
\cdot h^{- (3+ \frac{1}{2k})} \\ 
& \cdot &  \left(\frac{\ell}{L}\right)^{1/(2k)},
\end{eqnarray*}
since $\sum_{j=0}^{t_0+b-1} 2^j \leq 2^{t_0 + b} \leq 24 \pi\cdot (Cd)^{\frac{2-\eta}{1-\eta}} \cdot h^{-3}$.
Denote by $c(k)$ the constant on the expression above that depends on $k$, namely,
 $c(k) :=  \frac{80e}{3} \cdot \left(\frac{25}{3}\right)^{k-1}\cdot \left(\frac{8e}{3}\right)^{1/(2k)} \cdot (24\pi)^{1 + \frac{1}{2k}}$.
We have
\begin{eqnarray*}
N = \sum_{\ell =1}^L N_\ell &\leq & c(k)  \cdot C^{\frac{2-\eta}{1-\eta} + \frac{5-2\eta}{2k(1-\eta)}}\cdot d^{\frac{2-\eta}{1-\eta}+ \frac{3}{2k(1-\eta)}}
\cdot h^{- (3+ \frac{1}{2k})} \cdot \sum_{\ell = 1}^L\left(\frac{\ell}{L}\right)^{1/(2k)} \\
&\leq & c(k) \cdot C^{\frac{2-\eta}{1-\eta} + \frac{5-2\eta}{2k(1-\eta)}}\cdot d^{\frac{2-\eta}{1-\eta}+ \frac{3}{2k(1-\eta)}}
\cdot h^{- (3+ \frac{1}{2k})} \cdot L \\
&\leq & c(k) \cdot C^{\frac{4-2\eta}{1-\eta} + \frac{5-2\eta}{2k(1-\eta)}}\cdot d^{\frac{4-2\eta}{1-\eta}+ \frac{3}{2k(1-\eta)}}
\cdot h^{- (3+ \frac{1}{2k})} .
\end{eqnarray*}
For \textit{relative error} $O(\e)$, it suffices to set $h\leq \e$. In that case
\begin{equation}
N \leq c(k) \cdot \e^{-(3+\frac{1}{2k})} \cdot C^{\frac{4-2\eta}{1-\eta} + \frac{5-2\eta}{2k(1-\eta)}}\cdot d^{\frac{4-2\eta}{1-\eta}+ \frac{3}{2k(1-\eta)}}.
\end{equation}

Recall that the eigenvalues and eigenvectors of the
discretized Laplacian are known and the evolution of a system with
a Hamiltonian involving $-\Delta_h$ can be implemented with $d \cdot O(\log^2 \e^{-1})$
quantum operations using the Fourier transform in each dimension; see e.g.,
\cite[p. 209]{NC00}. The evolution of a system with a Hamiltonian involving
$V_h$ can be implemented using two quantum queries and phase kickback.
Hence the number of quantum operations required to implement
the algorithm is proportional to
\[c(k) \cdot \e^{-(3+\frac{1}{2k})} \cdot C^{\frac{4-2\eta}{1-\eta} + \frac{5-2\eta}{2k(1-\eta)}}\cdot d^{1+\frac{4-2\eta}{1-\eta}+ \frac{3}{2k(1-\eta)}}.\]

The analysis leads to the following theorem.

\begin{theorem}
Consider the ground state energy estimation problem
for the time-independent Schr\"o\-dinger
equation \eqref{TISE1},\eqref{TISE2}.
The quantum algorithm that applies 
$L = (Cd)^{(2-\eta)/(1-\eta)}$ stages of repeated phase estimation with:
\begin{itemize}
\item \textit{Number of qubits:}
The top register has $q:=b+t_0= 3 \log \e^{-1} + \frac{2 - \eta}{1-\eta}\cdot \log (Cd) + O(1)$ qubits, while
the bottom register has $\Theta \left(d \log_2 \e^{-1}\right)$ qubits.
\item \textit{Input state:} The top register
is initialized to $\ket{0}^{\otimes q}$.
The bottom register of the first stage 
is initialized to 
$\ket{u_{0,0}}$.
Furthermore we set
$\ket{\psi_{{\rm in},\ell}} := \ket{\psi_{{\rm out},\ell-1}}$
for $\ell = 2,3,\ldots ,L$.
\item \textit{Implementation of exponentials:}
Implement each exponential $W_{h,\ell}^{2^j}$
using Suzuki's splitting
formulas of order $2k+1$ with simulation error
$\e^{S}_{j,\ell} = 2^{j-q} \cdot (Cd)^{-2/(1-\eta)}$, for 
$j=0,1,\ldots ,q-1$, $\ell = 1,2,\ldots ,L$.
\end{itemize}
approximates the ground state energy $E_0$
with relative error $O(\e)$, as $d\e\to 0$, using
a number of bit queries proportional to
\[
c(k) \cdot \e^{-(3+\frac{1}{2k})} \cdot C^{\frac{4-2\eta}{1-\eta} + \frac{5-2\eta}{2k(1-\eta)}}\cdot d^{\frac{4-2\eta}{1-\eta}+ \frac{3}{2k(1-\eta)}}
\]
and a number of quantum operations proportional to
\[
c(k) \cdot \e^{-(3+\frac{1}{2k})} \cdot C^{\frac{4-2\eta}{1-\eta} + \frac{5-2\eta}{2k(1-\eta)}}\cdot d^{1+\frac{4-2\eta}{1-\eta}+ \frac{3}{2k(1-\eta)}},
\]
where $c(k) :=  \frac{80e}{3} \cdot \left(\frac{25}{3}\right)^{k-1}\cdot \left(\frac{8e}{3}\right)^{1/(2k)} \cdot (24\pi)^{1 + \frac{1}{2k}}$,
with constant probability of success.
\end{theorem} 
Finally, observe that at the end of the algorithm the  state in the bottom register $\ket {\psi_{\rm out,L}}$ 
satisfies
\[|\inner{\psi_{\rm out,L}}{u_{0,L}}|^2 \geq 1 - O\left( (Cd)^{-\frac{2-\eta}{1-\eta}}\right),\]
where $\ket{u_{0,L}}$  denotes the ground state eigenvector of $M_h$. This observation motivates the algorithm in the next
section.

\subsection{An algorithm preparing a quantum state approximating the ground state eigenvector}
\label{sec:ground_state}

In Section \ref{sec:eiv} we exhibited a quantum algorithm, based on repeated applications of phase estimation,
approximating the ground state energy
of the Hamiltonian $H$ (eq.~(\ref{TISE1}) and (\ref{TISE2})). It turns out we can use the same algorithm 
with different values of its parameters, to prepare a quantum state 
overlapping with the ground state eigenvector of the discretized Hamiltonian $M_h$.

In particular in this section we show an algorithm that:
\begin{enumerate}
\item estimates the ground state energy of the Hamiltonian $H$ with relative error $\e$,
\item approximates the ground state $\ket{u_{0,L}}$ of the discretized Hamiltonian $M_h$ 
by a state $\ket{\psi}$ such that
\[|\inner{u_{0,L}}{\psi}|^2 \geq 1 - O(\delta),\]
where $\delta \in (0,1)$. \footnote{In this section $\delta$ is an input parameter
to the algorithm, and is slightly different from $\delta$ as used in sections \ref{sec:prelim}--\ref{sec:succ_prob1}}
\end{enumerate}

We remark that for $\delta = \Omega \left((C d)^{-\frac{2-\eta}{1-\eta}}\right)$ we can use the algorithm in Section \ref{sec:eiv}.
From now on we assume that $\delta =  o \left((C d)^{-\frac{2-\eta}{1-\eta}}\right)$.

\subsubsection{Error analysis} 

We work exactly like in Section \ref{sec1:Error} to get the same number of
qubits $b$ for the top register of phase estimation (see (\ref{eq:b})), namely,
\[
b = \left\lceil \log \frac{2R\pi}{dh} \right \rceil = \left\lceil \log (6\pi h^{-3}) \right\rceil = \log \Theta \left(h^{-3}\right).
\]

\subsubsection{Success probability}

Equations (\ref{eq:prob_first}), (\ref{eq:t0}) remain the same. What changes is the number of 
stages $L$. Since we require 
\[|\inner{\psi_{\rm out, L}}{u_0}|^2 \geq 1 - O(\delta),\] 
we set 
\begin{equation}
\label{sec2:eq:L}
(Cd/L)^{2-\eta}= \delta \Rightarrow L = Cd / \delta^{1/(2-\eta)}.
\end{equation}
Just as before, the success probability of the algorithm after $L$ stages is
\[
P_{\rm total} \geq  \left( 1 - \frac{(\kappa_2 \cdot (Cd)^{2-\eta}/ L^{1-\eta})^2}{L} \right) \cdot  e^{-\kappa_2 \cdot  (Cd)^{2-\eta} / L ^{1-\eta}} = \left(1 - \frac{o(1)}{L}\right) \cdot e^{-o(1)} \geq 3/4,
\]
according to our choice of $L$ in (\ref{sec2:eq:L}) and the fact that $\delta = o\left((Cd)^{- \frac{2-\eta}{1-\eta}}\right)$,
since for larger $\delta$ we can use the algorithm in Section \ref{sec:eiv} as we pointed out.

\subsubsection{Matrix exponential  error}

Just like before, we approximate $W_{h,\ell}^{2^j}$ with error
\[
\e^{S}_{j,\ell} := 2^{j-(b+t_0)}/(L/Cd)^2,
\]
which according to our choice of $L$ becomes
\[
\e^{S}_{j,\ell} = 2^{j-(b+t_0)} \delta^{2/(2-\eta)}.
\] 
Note that the total error of phase estimation at each stage is
$2 \cdot \sum_{j = 0}^{b+t_0 - 1} \e^{S}_{j,l}  = 2\cdot \delta^{2/(2-\eta)}$, 
which is asymptotically smaller than
$\left(\frac{Cd}{L}\right)^{2-\eta} = \delta$.

\subsubsection{Cost}

We work as before. Using the bounds on the number of
exponentials \cite{PZ12} required to simulate $W_{h,\ell}^{2^j}$ we have
\[
N_{\ell} = \sum_{j=0}^{t_0 + b -1 }N_{j,\ell} \leq \sum_{j=0}^{t_0 + b -1} 2 \cdot 5 \cdot 5^{k-1} \cdot \norm{\Delta_h/R}_2 \cdot 
2^{j} \left( \frac{4e\cdot 2 \cdot 2^j \norm{\frac{\ell}{L} \cdot \frac{V_h}{R}}_2}{\e^{S}_{j,\ell}}
 \right)^{1/(2k)}\cdot \frac{4e\cdot 2}{3} \left(\frac{5}{3}\right)^{k-1},
\]
where the order of the splitting formula is $2k+1$, $k\ge1$.
Since $\norm{\Delta_{h}/R}_2 \leq 1$ and $\norm{\frac{\ell}{L} \cdot \frac{V_h}{R}} \leq \frac{\ell}{L} \cdot \frac{C}{3dh^{-2}} $ we have
\begin{eqnarray*}
N_{\ell} & \leq & 
\frac{80e}{3} \cdot 5^{k-1} \cdot \left(\frac{8e}{3}\right)^{1/2k} \cdot \left(\frac 5 3 \right)^{k-1}
\cdot 2^{(b+t_0)/(2k)} \cdot \left(\frac \ell L \right)^{1/2k} \cdot \left(\frac{C/(dh^{-2})}{\delta^{2/(2-\eta)}}\right)^{1/(2k)} \\ 
&\cdot & \sum_{j=0}^{t_0 + b -1} 2^j 
\\
 & \leq &
\frac{80e}{3} \cdot 5^{k-1} \cdot \left(\frac{8e}{3}\right)^{1/2k} \cdot \left(\frac 5 3 \right)^{k-1}
\cdot 2^{\frac{b+t_0}{2k}\left(1 + \frac{1}{2k}\right)}\cdot \left(\frac \ell L \right)^{1/(2k)} 
\cdot C^{1/(2k)} \cdot d^{-1/(2k)} \\ 
&\cdot & h^{1/k} 
\cdot \delta^{-\frac{1}{k(2-\eta)}} 
\\
& \leq & 
\frac{80e}{3} \cdot 5^{k-1} \cdot \left(\frac{8e}{3}\right)^{1/2k} \cdot \left(\frac 5 3 \right)^{k-1}
\cdot (24\pi)^{1+\frac{1}{2k}}
\cdot \left(\frac \ell L \right)^{1/(2k)}
\cdot C^{1/(2k)} \cdot d^{-1/(2k)} \\
& \cdot &  h^{-\left(3+\frac{1}{2k} \right)}
\cdot \delta^{-\frac{1}{k(2-\eta)}},
\end{eqnarray*}
since $2^{b} \leq 12\pi h^{-3}$ and $2^{t_0} \leq 2 \delta^{-1}$.
Once again, let $c(k) : = 
\frac{80e}{3} \cdot \left(\frac{25}{3}\right)^{k-1} \cdot \left(\frac{8e}{3}\right)^{1/(2k)} \cdot (24\pi)^{1+\frac{1}{2k}}$.
The total number of exponentials required is 
\begin{eqnarray*}
N = \sum_{\ell=1}^L N_\ell 
& \leq &  
c(k) \cdot C^{1/(2k)} \cdot d^{-1/(2k)} \cdot h^{-\left(3+\frac{1}{2k} \right)}
\cdot \delta^{-\frac{1}{k(2-\eta)}} \cdot
\sum_{\ell=1}^L \left(\frac \ell L \right)^{1/(2k)} 
\\
& \leq &
c(k)\cdot C^{1/(2k)} \cdot d^{-1/(2k)} \cdot h^{-\left(3+\frac{1}{2k} \right)}
\cdot \delta^{-\frac{1}{k(2-\eta)}} \cdot L
\\
&\leq &
c(k)\cdot C^{1+\frac{1}{2k}} \cdot d^{1-\frac{1}{2k}} \cdot h^{-\left(3+\frac{1}{2k} \right)}
\cdot \delta^{ - 1 - \frac{1}{2k} - \frac{1}{2-\eta} - \frac{1}{k(2-\eta)}} 
\end{eqnarray*}
For \textit{relative error} $O(\e)$, it suffices to set $h \leq \e$. In that case
\begin{equation}
N \leq c(k)\cdot C^{1+\frac{1}{2k}} \cdot d^{1-\frac{1}{2k}} \cdot \e^{-\left(3+\frac{1}{2k} \right)}
\cdot \delta^{ - 1 - \frac{1}{2k} - \frac{1}{2-\eta} - \frac{1}{k(2-\eta)}} 
\end{equation}

The analysis above leads to the following theorem.

\begin{theorem}
Consider the ground state energy and ground state eigenvector estimation problem
for the time-independent Schr\"odinger
equation \eqref{TISE1},\eqref{TISE2} with a convex potential. Let $\delta = o((Cd)^{-\frac{2-\eta}{1-\eta}})$.
The quantum algorithm that applies $L = Cd \cdot \delta^{-1/(2-\eta)}$ stages 
of repeated phase estimation with
\begin{itemize}
\item \textit{Number of qubits:}
The top register has $q :=b+t_0= 3 \log \e^{-1} +  \log \delta^{-1} + O(1)$ qubits, and
the bottom register has $\Theta (d \log \e^{-1})$ qubits.
\item \textit{Input state:} The top register
is always initialized to $\ket{0}^{\otimes q}$.
The bottom register in the first stage 
is initialized to 
$\ket{u_{0,0}}$.
Thereafter, the bottom register is set according to
$\ket{\psi_{{\rm in},\ell}} := \ket{\psi_{{\rm out},\ell-1}}$
for $\ell = 2,3,\ldots ,L$.
\item \textit{Implementation of exponentials:}
Implement each exponential $W_{h,l}^{2^j}$
using Suzuki splitting
formulas of order $2k+1$ with simulation error
$\e^{S}_{j,l} = 2^{j-q}\cdot \delta^{2/(2-\eta)} $, for 
$j=0,1,\ldots ,q-1$.
\end{itemize}
approximates the ground state energy $E_0$
with relative error $O(\e)$, for $\e = o(d^{-2})$, using a number of bit queries
proportional to
\[
c(k)\cdot C^{1+\frac{1}{2k}} \cdot d^{1-\frac{1}{2k}} \cdot \e^{-\left(3+\frac{1}{2k} \right)}
\cdot \delta^{ - 1 - \frac{1}{2k} - \frac{1}{2-\eta} - \frac{1}{k(2-\eta)}}
\]
and a number of quantum operations proportional to
\[
c(k)\cdot C^{1+\frac{1}{2k}} \cdot d^{2-\frac{1}{2k}} \cdot \e^{-\left(3+\frac{1}{2k} \right)}
\cdot \delta^{ - 1 - \frac{1}{2k} - \frac{1}{2-\eta} - \frac{1}{k(2-\eta)}} ,
\]
where $c(k) : = 
\frac{80e}{3} \cdot 5^{k-1} \cdot \left(\frac{8e}{3}\right)^{1/2k} \cdot \left(\frac 5 3 \right)^{k-1}
\cdot (24\pi)^{1+\frac{1}{2k}}$.
The 
final state on the lower register $\ket{\psi_{\rm out,L}}$ satisfies
\[|\inner{u_0}{\psi_{\rm out,L}}|^2 \geq 1 - O\left( \delta\right),\]
where $\ket{u_0}$ is the ground state eigenvector of $M_h$.
The algorithm succeeds with probability at least $3/4$.
\end{theorem} 

\section{Acknowledgements}

The authors would like to thank Joseph F. Traub for useful comments and suggestions.
This research has been supported in part by NSF/DMS.

\newpage

\section{Appendix}

We derive some useful intermediate results for
the analysis in Section \ref{sec:phase_estimation}.

\begin{prop} 
\label{thm1}
Consider phase estimation with initial state $\ket{0}^{\otimes b}\otimes \ket{\psi_{\rm in}}$ 
and the unitaries $\tilde U_t$, $t=0,1,\ldots ,2^b-1$. Let $m$ be 
the measurement outcome of phase estimation and let
$\ket{\psi_m}$ be the state after the measurement in the bottom register.
Let $c_0=\inner{\psi_{\rm in}}{u_0}$ and $c_0^{\prime}=\inner{\psi_m}{u_0}$, where $\ket{u_0}$ is  ground state eigenvector.
If
\begin{itemize}
\item $b$ is such that the phases  satisfy $\left|\phi_j - \phi_0\right| > \frac{5}{2^{b}}$ 
for all $j=1,2,\ldots ,n^d-1$.
\item $|c_0|^2 \geq \frac{\pi^2}{16}$.
\end{itemize}
Then, with probability $p \geq |c_0|^2 \cdot \frac{4}{\pi^2} - 2 \sum_{j=0}^{b-1} \norm{U^{2^j} - \tilde U_{2^j}}$, we get an outcome $m$ such that
\begin{itemize}
\item
$\left|\phi_0 - \frac{m}{2^b}\right| \leq \frac{1}{2^{b+1}}$
\end{itemize}
and
\begin{itemize}
\item if $1-|c_0|^2 \leq \gamma \e_H$ then $1 - |c_0^{\prime}|^2 \leq 
(\gamma + 14)\e_H$
\item if $1-|c_0|^2 \geq \gamma \e_H^{1-\eta}$, for $\eta \in (0,1)$, 
then $|c_0^{\prime}| \geq |c_0|$
\end{itemize}
where $\gamma$ is a positive constant. 
\end{prop}

\begin{proof} After the application of $H^{\otimes b}$ on the top in phase estimation
 the state becomes
\[\frac{1}{2^{b/2}}\sum_{k=0}^{2^b -1 }\ket{k} 
\sum_{j=0}^{n^{d}-1}c_j\ket{u_j}\]
The state of the system after the application of the controlled $\tilde U_{2^t}$,
$t=0,1,\ldots ,b-1$ is 
\[\sum_{j=0}^{n^d-1}c_j \frac{1}{2^{b/2}} \sum_{k=0}^{2^b-1} \ket{k} 
\tilde{U}_k\ket{u_j}
= \frac{1}{2^{b/2}} \sum_{j=0}^{n^d-1}c_j \sum_{k=0}^{2^b-1} \ket{k} 
\left(U^k\ket{u_j} + D_k \ket{u_j}\right),\]
where $D_k = \tilde{U}_k - U^k$. Then $\norm{D_k} \leq \e_H$. 
Since 
$\ket{u_j}$, $j = 0,1,\ldots , n^d-1$,
are the eigenvectors of $U$, the state can be written as $\ket{\psi_1} + \ket{\psi_2}$, where
\[
\ket{\psi_1} = \frac{1}{2^{b/2}} \sum_{j=0}^{n^d-1}c_j \sum_{k=0}^{2^b-1} \ket{k} U^k\ket{u_j} = 
\frac{1}{2^{b/2}} \sum_{j=0}^{n^d-1}c_j \sum_{k=0}^{2^b-1} \ket{k} e^{2\pi i k\phi_j}\ket{u_j} ,
\]
and 
\[
\ket{\psi_2} =  \frac{1}{2^{b/2}} \sum_{j=0}^{n^d-1} c_j  
\sum_{k=0}^{2^b-1}\ket{k} D_k \ket{u_j} = 
\frac{1}{2^{b/2}} \sum_{j=0}^{n^d-1} c_j  \sum_{k=0}^{2^b-1}\ket{k} 
\ket{x_{j,k}} ,
\]
where 
$\ket{x_{j,k}} := D_k \ket{u_j}$. Clearly $\norm{\ket{x_{j,k}}}\leq \e_H$, 
for all $k = 0,1, \ldots ,2^b -1$ and $j = 0,1,\ldots , n^b-1$.

The next step in phase estimation is to apply 
$\mathbb{F}^{H}\otimes I$, where $\mathbb{F}^H$ is the inverse
Fourier transform.
The state becomes $\ket{\psi_{\mathbb{F}^{H}}} = 
\ket{\psi_{1,\mathbb{F}^{H}}} + \ket{\psi_{2,\mathbb{F}^{H}}}$, where
\begin{eqnarray*}
\ket{\psi_{1,\mathbb{F}^{H}}} 
&=& \frac{1}{2^{b/2}} \sum_{j=0}^{n^d-1}c_j \sum_{k=0}^{2^b-1} \mathbb{F}^{H}\ket{k}  e^{2\pi i k\phi_j} \ket{u_j} \\ 
&=& \frac{1}{2^b} \sum_{j=0}^{n^d-1}c_j \sum_{k,\ell=0}^{2^b-1} e^{2\pi i k\left(\phi_j - \frac{\ell}{2^b}\right)} \ket{\ell} \ket{u_j} \\ 
&=& \sum_{j=0}^{n^d-1}c_j \sum_{\ell=0}^{2^b-1} \alpha(\ell,\phi_j) \ket{\ell} \ket{u_j}
\end{eqnarray*}
where 
\begin{equation}\label{eq:alphas}
\alpha (\ell,\phi_j) := \frac{1}{2^b} \sum_{k=0}^{2^b-1} e^{2\pi i k\left(\phi_j - \frac{\ell}{2^b}\right)}
\end{equation}
and
\[
\ket{\psi_{2,\mathbb{F}^{H}}} = \sum_{j=0}^{n^d-1} c_j \frac{1}{2^{b/2}} 
\sum_{k=0}^{2^b-1}\mathbb{F}^{H}\ket{k} \ket{x_{j,k}} = 
\sum_{j=0}^{n^d-1} c_j \frac{1}{2^b} \sum_{k,\ell=0}^{2^b-1} e^{-2\pi i \ell k/2^b } 
\ket{\ell} \ket{x_{j,k}}.
\]

Finally we measure the top register on the computational basis, and denote the outcome by $m$. The resulting state is 
$\ket{m,\psi_m} = \frac{\ket{m,\psi_{1,m}} + \ket{m,\psi_{2,m}}}{\norm{\ket{m,\psi_{1,m}} + \ket{m,\psi_{2,m}}}}$ where
\[
\ket{m,\psi_{1,m}} = 
\frac{1}{2^b} \sum_{j=0}^{n^d-1}c_j \sum_{k=0}^{2^b-1} 
e^{2\pi i k\left(\phi_j - \frac{m}{2^b}\right)} \ket{m}\ket{u_j} = 
\sum_{j=0}^{n^d-1}c_j \alpha (m,\phi_j) \ket{m}\ket{u_j},
\]
and
\begin{equation} \label{eq:psi2}
\ket{m,\psi_{2,m}}  =
 \frac{1}{2^b} \sum_{j=0}^{n^d-1} c_j \sum_{k=0}^{2^b-1} e^{-2\pi i m k/2^b } \ket{m} \ket{x_{j,k}}
\end{equation}

We now consider the magnitude of the projection of the resulting state $\ket{m,\psi_m}$
on $\ket{m,u_0}$, namely $|c_0^{\prime}| = \left| \inner{m,u_0}{m,\psi_m}\right| =\left| \inner{u_0}{\psi_m} \right|$.
 
We have
\begin{eqnarray}
|  c_0^\prime|^2 = \frac{\left| c_0 \alpha (m,\phi_0) + \frac{1}{2^b} \sum_{j=0}^{n^d-1} c_j \sum_{k=0}^{2^b-1} e^{-2\pi i m k/2^b }  \inner{u_0}{x_{j,k}} \right|^2}{\norm{\ket{\psi_{1,m}} + \ket{\psi_{2,m}}}^2} 
\label{eq:1}
\end{eqnarray} 

If we did not have the matrix exponential approximation error, then
with probability at least $|c_0|^2\cdot 8/\pi^2$ the measurement yields an outcome $m$ such that
$\left|\phi_0 - \frac{m}{2^b}\right| < \frac{1}{2^{b}}$, see \cite[Thm. 11]{BHMT02}.
There are at most two such outcomes. The one that has the highest probability 
among them leads to an estimate that is closest to the phase.
From now on let $m$ denote the outcome for which $m/2^b$ is closest to the phase. Then its
probability is $|c_0|^2 \cdot |\alpha(m, \phi_0)|^2 \geq |c_0|^2 \cdot \frac{4}{\pi^2}$
and the error satisfies
$\left|\phi_0 - \frac{m}{2^b}\right| \leq \frac{1}{2^{b+1}}$ .

If we account for the error, the probability of $m$  becomes at least
$|c_0|^2 \cdot \frac{4}{\pi^2} - 2 \sum_{j=0}^{b-1} \norm{U^{2^j} - \tilde U_{2^j}}$; see
\cite[pg. 195]{NC00}.

From Lemma \ref{Lem:2} we have $\norm{\ket{\psi_{2,m}}} \leq \e_H$. 
In addition $\inner{m,\psi_{2,m}}{m,u_0} \leq \norm{\ket{\psi_{2,m}}}\cdot \norm{\ket{u_0}} \leq \e_H$. 
Using Lemma \ref{Lem:1} and for $\e_H < \sqrt 8 / \pi^2$, equation (\ref{eq:1}) becomes

\begin{eqnarray}
|c_0^{\prime}| &>& |c_0|  \left(\frac{|\alpha (m,\phi_0)|}{\sqrt{\sum_{j=0}^{n^d-1} |c_j|^2 \cdot |\alpha (m,\phi_j)|^2 }}- 7 \frac{\e_H}{|c_0|}\right) \nonumber \\
&=& |c_0| \left(\frac{1}{\sqrt{\sum_{j=0}^{n^d-1} |c_j|^2 \cdot \frac{|\alpha (m,\phi_j)|^2}{|\alpha (m, \phi_0)|^2} }}- 7 \frac{\e_H}{|c_0|} \right) \nonumber \\
&=& |c_0| \left(\frac{1}{\sqrt{|c_0|^2 + \sum_{j=1}^{n^d-1} |c_j|^2 \cdot \frac{|\alpha (m,\phi_j)|^2}{|\alpha (m, \phi_0)|^2} }}- 7 \frac{\e_H}{|c_0|} \right) 
 \label{eq:3}
\end{eqnarray}

Let $k$ be such that 
$| \alpha(m,\phi_k) |^2 = \max_{i\ge 1} |\alpha(m,\phi_i)|^2$. Then using the assumption $|\phi_k -\phi_0|>5/2^b$ 
we have that $|\phi_k - m/2^b |> 4/2^b$. From \cite[Thm. 11]{BHMT02} we get that
$|\alpha(m,\phi_k)|^2 \leq \frac{1}{(2\cdot 2^{b}\cdot 2^{-b+2})^2} = \frac{1}{64}$. 
Combine this with  $|\alpha(m,\phi_0)|^2 \geq \frac{8}{2\pi^2} = \frac{4}{\pi^2}$ to obtain
\begin{eqnarray}
|c_0^{\prime}| &>& |c_0| \left(\frac{1}{\sqrt{|c_0|^2 + \frac{\pi^2}{256}\cdot \sum_{j=1}^{n^d-1} |c_j|^2 }}- 7 \frac{\e_H}{|c_0|} \right) \nonumber \\
 &=& |c_0| \left(\frac{1}{\sqrt{|c_0|^2 + \frac{\pi^2}{256}\cdot (1 - |c_0|^2) }}- 7 \frac{\e_H}{|c_0|} \right) ,
 \label{eq:4}
\end{eqnarray}
since $\sum_{j=0}^{n^d-1}|c_j|^2 = 1$.

Now examine the different cases, depending on the
magnitude of $|c_0|$.

\textit{Case 1:} $1-|c_0|^2 \leq \gamma \e_H$, for a constant $\gamma$. 
Then (\ref{eq:4}) becomes
\[
|c_0^{\prime}| > |c_0| \left(1 - 7 \frac{\e_H}{\sqrt{1-\gamma \e_H}} \right) ,
\] 
because $f(x) = \frac{1}{\sqrt{x + \frac{\pi^2}{256} (1 - x)}}$ is a monotonically decreasing function
for $x \in [0,1]$. Hence
\begin{eqnarray*}
|c_0^{\prime}|^2 &>& |c_0|^2 \left(1 - \frac{14}{\sqrt{1-\gamma \e_H}} \e_H + \frac{49}{1-\gamma \e_H} \e_H^2\right) \\
&\geq& (1-\gamma \e_H) \cdot \left(1 - \frac{14}{\sqrt{1-\gamma \e_H}} \e_H + \frac{49}{1-\gamma \e_H} \e_H^2\right) \\
&=& 1 - \gamma \e_H - 14 \e_H \sqrt{1-\gamma \e_H} + 49 \e_H^2 \\
&\geq & 1 - \gamma \e_H - 14 \e_H + 49 \e_H^2 \geq 1 - (\gamma + 14) \e_H ,
\end{eqnarray*}
where the second from last inequality holds because $1- \gamma \e_H < 1$. This concludes the first part of the theorem.

\textit{Case 2:} $1-|c_0|^2 \geq \gamma \e_H^{1-\eta}$, for $\eta \in (0,1)$ and $\gamma >0$.
Then (\ref{eq:4}) becomes
\[
|c_0^{\prime}| > |c_0| \left(\frac{1}{\sqrt{1- \left(1 - \frac{\pi^2}{256}\right) \gamma \e_H^{1-\eta}}} - 7 \frac{\e_H}{\pi/ 4} \right),
\] 
because $f(x) = \frac{1}{\sqrt{x + \frac{\pi^2}{256} (1 - x)}}$ is a
monotonically decreasing function for $x \in [0,1-\gamma \e^{1-\eta}_H]$
and $|c_0|^2 \geq \frac{\pi^2}{16}$.

Note that $\frac{1}{\sqrt{1-a}} \geq \sqrt{ 1 + a }$, for $ |a| \leq 1$. Hence
\begin{eqnarray*}
|c_0^{\prime}|^2 &>& |c_0|^2 \left( 1 + \left(1 - \frac{\pi^2}{256}\right)\gamma \e_H^{1-\eta} - \frac{56}{\pi} \e_H \sqrt{1 + \left( 1 - \frac{\pi^2}{256}\right)\gamma \e_H^{1-\eta}} + \frac{28^2}{\pi^2} \e_H^2\right) \\
&>& |c_0|^2 \left(1 + \left(1 - \frac{\pi^2}{256}\right)\gamma \e_H^{1-\eta} - O(\e_H)\right)
> |c_0|^2,
\end{eqnarray*}
for $\e_H$ sufficiently small.
\end{proof}

\begin{lem} 
\label{Lem:1}
For $0 \leq \e_H < \frac{\sqrt 8}{\pi^2}$, $|c_0|^2\geq \frac{\pi^2}{16}$ 
and assuming the measurement outcome $m$ in phase estimation satisfies 
$\left|\frac{m}{2^b}-\phi_0\right|\leq 2^{-(b+1)}$
we have
\[
\frac{|\al(m,\phi_0)-\e_H|}{\sqrt{\sum_{j=0}^{n^d-1}|c_j|^2 |\al(m,\phi_j)|^2} + \e_H} > \frac{|\al(m,\phi_0)|}{\sqrt{\sum_{j=0}^{n^d-1}|c_j|^2 |\al(m,\phi_j)|^2}} - 7 \e_H ,
\]
where $c_0$ and $a(m,\phi_0)$ are defined in \eqref{cj} and \eqref{eq:alphas}, respectively.
\end{lem}
\begin{proof}
We first show $\frac{|\al(m,\phi_0)|}{\sqrt{\sum_{j=0}^{n^d-1}|c_j|^2 |\al(m,\phi_j)|^2} +
 \e_H} > \frac{|\al(m,\phi_0)|}{\sqrt{\sum_{j=0}^{n^d-1}|c_j|^2 |\al(m,\phi_j)|^2}} - \gamma 
 \e_H$ for $\gamma > 4 $. 
We have
\begin{eqnarray*}
\frac{|\al(m,\phi_0)|}{\sqrt{\sum_{j=0}^{n^d-1}|c_j|^2 |\al(m,\phi_j)|^2} + \e_H} &>& \frac{|\al(m,\phi_0)|}{\sqrt{\sum_{j=0}^{n^d-1}|c_j|^2 |\al(m,\phi_j)|^2}} 
- \gamma \e_H \\
\Leftrightarrow \gamma &>& \frac{|\al(m,\phi_0)|}{\sqrt{\sum_{j=0}^{n^d-1}|c_j|^2 |\al(m,\phi_j)|^2} (\sqrt{\sum_{j=0}^{n^d-1}|c_j|^2 |\al(m,\phi_j)|^2} + \e_H) }
\end{eqnarray*} 

Then
$\sqrt{\sum_{j=0}^{n^d-1}|c_j|^2 |\al(m,\phi_j)|^2} \geq |\al(m,\phi_0)||c_0| \geq \frac{1}{2}$, 
since $\frac{\sqrt 4}{\pi}\leq |\al(m,\phi_0)| \leq 1$, because $\left|\frac{m}{2^b}-\phi_0\right|\leq 2^{-(b+1)}$ \cite{BHMT02}, and $|c_0|\geq \pi/4$.
Hence
\begin{eqnarray*}
\frac{|\al(m,\phi_0)-\e_H|}{\sqrt{\sum_{j=1}^{n^d-1}|c_j|^2 |\al(m,\phi_j)|^2} + \e_H} 
&>& \frac{|\al(m,\phi_0)|}{\sqrt{\sum_{j=1}^{n^d-1}|c_j|^2 |\al(m,\phi_j)|^2} + \e_H} - 
2 \e_H \\
&>& \frac{|\al(m,\phi_0)|}{\sqrt{\sum_{j=1}^{n^d-1}|c_j|^2 |\al(m,\phi_j)|^2}} - 
(\gamma + 2 ) \e_H ,
\end{eqnarray*}
for $\gamma > 4 $. Taking $\gamma = 5$  completes the proof.
\end{proof}

\begin{lem}
\label{Lem:2} Consider $\ket {m,\psi_{2,m}}$ as defined in \eqref{eq:psi2}. Then
$\norm{\ket{\psi_{2,m}}} \leq \e_H$.
\end{lem}

\begin{proof}
We have
\begin{eqnarray*}
\norm{\ket{\psi_{2,m}}}  &=&
 \norm{ \frac{1}{2^b} \sum_{j=0}^{n^d-1} c_j \sum_{k=0}^{2^b-1} e^{-2\pi i m k/2^b } \ket{m} \ket{x_{j,k}}} \\ 
 &\leq & \frac{1}{2^b} \sum_{k=0}^{2^b-1} \left|e^{-2\pi i m k/2^b }\right| \norm{\ket m } \cdot \norm{\sum_{j=0}^{n^d-1} c_j  \ket{x_{j,k}}} \\
 &=& \frac{1}{2^b} \sum_{k=0}^{2^b-1}  \norm{ D_k \sum_{j=0}^{n^d-1} c_j \ket{u_j}} \\
 &\leq & \frac{1}{2^b} \sum_{k=0}^{2^b-1}  \norm{ D_k } \cdot \norm{ \sum_{j=0}^{n^d-1} c_j \ket{u_j}} \\
 &=& \frac{1}{2^b} \sum_{k=0}^{2^b-1}  \norm{ D_k } \leq \e_H ,
\end{eqnarray*}
where $D_k=U_k-\tilde U_k$ and since $\norm{D_k} \leq \e_H$ due to (\ref{eq:expErr}), and $\norm{\sum_{j=0}^{n^d-1} c_j \ket{u_j}} =1$.
\end{proof}

\begin{lem} 
\label{lem:3}
Under the conditions of Theorem \ref{thm1mod}, i.e.
$m^\prime$ is the result of phase estimation for which 
$\left|\frac{m^\prime}{2^{t_0 + b}} - \phi_0\right| \leq \frac{1}{2^b}$,
$\left|\phi_j - \phi_0\right| > \frac{5}{2^{b}}$ 
for all $j=1,2,\ldots ,n^d-1$ and $|c_0|^2 \geq \pi^2 / 16$,
we have
\[
\frac{|\al(m^{\prime},\phi_j)|^2}{|\al(m^{\prime},\phi_0)|^2} \leq \frac{\pi^2}{32}
\]
with probability 
$p(t_0) \geq |c_0|^2\left( 1 -\frac 1{2(2^t_0-1)}\right)  - \left(\frac{5\pi^2}{2^5} + \frac{1-\frac{\pi^2}{16}}{2^5}\right) \cdot \frac{1}{2^{t_0}}$.
\end{lem}
\begin{proof}
Let $M_2 = 2^{t_0 + b}$ and $\Delta_j = \left| \frac{m^\prime}{M_2} - \phi_j\right|$, $i=0,\dots,n^d-1$. For $j\geq 1$ we have \cite{BHMT02} 
\begin{equation}
\label{eq:far}
|\al(m^{\prime},\phi_j)|^2 \leq \frac{1}{4 (M_2 \Delta_j)^2} \leq \frac{1}{2^{2t_0 + 6}}.
\end{equation}
Observe that $\Delta_j = \left| \frac{m^\prime}{M_2} - \phi_j\right| = 
\phi_j - \frac{m^\prime}{M_2} \geq \frac{4}{2^b}$, 
since $\phi_j > m^{\prime}/M_2$. Also 
\begin{equation}
\label{eq:al0}
\left|\al (m^\prime,\phi_0)\right|^2 = M_2^{-2}\cdot \frac{\sin^2 (M_2 \Delta_0 \pi)}{\sin^2 (\Delta_0 \pi)} \geq 
\frac{\sin^2 (M_2 \Delta_0 \pi)}{(M_2 \Delta_0 \pi)^2 } = \sinc^2 (M_2 \Delta_0 \pi) = 
\frac{\sin^2 (M_2 \phi_0 \pi)}{M_2^2 (\Delta_0 \pi)^2 } 
\end{equation}
where the inequality follows from $\Delta_0 \pi< \pi 2^{-b} \le \pi/2$. Hence,
\[
\frac{|\al(m^{\prime},\phi_j)|^2}{|\al(m^{\prime},\phi_0)|^2} \leq \frac{\pi^2}{4} \left(\frac{\Delta_0}{\Delta_j}\right)^2
\frac{1}{\sin^2 (M_2 \phi_0 \pi)}.
\]
Note that $\Delta_j \ge 4/2^b$ which yields
\begin{equation}
\label{eq:frac1}
\frac{|\al(m^{\prime},\phi_j)|^2}{|\al(m^{\prime},\phi_0)|^2} \le \frac{\pi^2}{2^6}\cdot \frac{1}{\sin^2 (M_2 \phi_0 \pi)}
\end{equation}

Note that the upper bound in (\ref{eq:frac1}) depends on how close 
$M_2 \phi_0$ is to an integer or, equivalently, on  is the
fractional part of $m^\prime - M_2 \phi_0$ for $m^\prime \in \mathcal{G}$.
We remark that  if $M_2\phi_0$ is an integer then $|\al(m^{\prime},\phi_0)|^2=1$  and the lemma statement holds trivially.

Consider, without loss of generality, the closest result $m_0$
to $M_2\phi_0$ such that $m_0 - M_2 \phi_0 = Y \cdot 2^{-q} \le 1/2$, $0<Y\le 1$, 
i.e.,  $m_0 > M_2 \phi_0$. (The case $m_0 < M_2\phi_0$ is dealt with similarly and we omit it.)
  
Denote by $m_\ell$, for $\ell = -2^{t_0}, -2^{t_0}+1, \ldots ,2^{t_0}-1$,
the measurement result such that $m_\ell - M_2\phi_0 = \ell + Y\cdot 2^{-q}$. These 
are all the elements of $\mathcal{G}$,  the set  defined in Theorem \ref{thm1mod}.

\textit{Case 1:} Let $1/2 \ge Y \cdot 2^{-q} \geq 1/4$. Then
\[
\sin^2 (M_2 \phi_0 \pi) = \sin^2( m_\ell\pi- M_2\phi_0\pi) 
= \sin^2(Y\cdot 2^{-q} \pi)
\geq \sin^2 (\pi/4) \geq 1/2,\]
which according to equation (\ref{eq:frac1}) implies
\begin{equation}
\label{eq:3.2}
\frac{|\al(m_\ell,\phi_j)|^2}{|\al(m_\ell,\phi_0)|^2} \leq \frac{\pi^2}{32},
\end{equation}
for all $m_\ell \in \mathcal{G}$.

\textit{Case 2:}
We now examine the case where $Y \cdot 2^{-q} < 1/4$, i.e. $q \geq 2$. 
In this case we deal with results $m_\ell$ in the set $\mathcal{G}$
for which the bound for $\frac{|\al(m_\ell,\phi_j)|^2}{|\al(m^{\prime},\phi_0)|^2}$
may become greater than $\frac{\pi^2}{32}$. We show that these results occur with 
probability at most $\left(\frac{5\pi^2}{2^5} + \frac{1-\frac{\pi^2}{16}}{2^5}\right)2^{-t_0}$.

Using  equation 
(\ref{eq:al0}) for $m_0$ we obtain
\[
\left|\al (m_0,\phi_0)\right|^2 \geq  \frac{\sin^2 ( Y\cdot 2^{-q} \pi)}{(Y \cdot 2^{-q} \pi)^2}
 = \sinc^2 (Y\cdot 2^{-q} \pi).
\]
Note that $Y\cdot 2^{-q} < 1/4$, hence $\sinc(\cdot)$ is decreasing. As a result
\[
\left|\al (m_0,\phi_0)\right|^2 \geq \sinc^2 (2^{-q} \pi)
\geq \left(\frac{2^{-q} \pi - \frac{(2^{-q} \pi)^3}{6}}{2^{-q} \pi}\right)^2
= \left(1 - \frac{\pi^2}{6 \cdot 2^{2q}}\right)^2,
\]
since $\sin(x) \geq x - \frac{x^3}{3!}$, for $x < 1$.
Furthermore, since $q\geq 2$ using equation (\ref{eq:far}) we get
\begin{equation}
\label{eq:m_0}
\frac{|\al(m_0,\phi_j)|^2}{|\al(m_0,\phi_0)|^2} \leq \frac{1}{4 \cdot 2^{2(t_0 + 2)}} \cdot \left(1 - \frac{\pi^2}{6 \cdot 2^{2q}}\right)^{-2}
<  \frac{1}{28} \cdot \left(1- \frac{\pi^2}{6\cdot 2^4}\right)< \frac{\pi^2}{32},
\end{equation}
since $t_0\geq 1$ and $q\geq 2$.

We now examine the remaining results $m_\ell$ for $\ell = \pm 1, \pm 2,\ldots , \pm (2^{t_0}-1), -2^{t_0}$.
From equation (\ref{eq:al0}) we have 
\[
\left|\al (m_\ell,\phi_0)\right|^2  \geq \sinc^2 ((\ell + Y \cdot 2^{-q}) \pi)
 = \frac{\sin^2 (Y \cdot 2^{-q} \pi)}{((\ell + Y \cdot 2^{-q}) \pi)^2}
 \geq \frac{2^{-2(q+1)} \cdot 8}{(\ell + Y \cdot 2^{-q})^2 \pi^2},
\]
since $1/4 > Y \cdot 2^{-q} \geq 2^{-(q+1)}$ and 
$\sin x \geq \frac{2\sqrt 2}{\pi} x$, for $x < \pi/4$.
From equation (\ref{eq:far}) we have
\[
\frac{|\al(m_\ell,\phi_j)|^2}{|\al(m_l,\phi_0)|^2} \leq \frac{\pi^2}{8} \cdot \frac{(\ell+Y \cdot 2^{-q})^2 \cdot 2^{2(q+1)}}{2^{2t_0+6}}
\]
which for $\ell \geq 1$ implies
\begin{equation}
\label{eq:upperplus}
\frac{|\al(m_\ell,\phi_j)|^2}{|\al(m_l,\phi_0)|^2} \leq \frac{\pi^2}{2^9} \cdot \frac{(\ell+1)^2}{2^{2t_0}} \cdot 2^{2(q+1)}  =: \beta (\ell,q,t_0)
\end{equation}
and for $\ell \leq -1$ implies
\begin{equation}
\label{eq:upperminus}
\frac{|\al(m_\ell,\phi_j)|^2}{|\al(m_l,\phi_0)|^2} \leq \frac{\pi^2}{2^9} \cdot 
\frac{\ell^2}{2^{2t_0}} \cdot 2^{2(q+1)}  =: \beta (\ell,q,t_0)
\end{equation}

If $q$ is relatively large, although $M_2 \phi_0$ is very close to $m_0$,  
we might have $\beta (\ell,q,t_0) > \pi^2/32$, for some results $m_\ell$. Let 
$\mathcal{B} = \{\ell \in \{-2^{t_0}, -2^{t_0}+1, \ldots , 2^{t_0}-1\}: 
\beta(\ell,q,t_0) > \pi^2/32\}$ 
the set of the indices of those results, with $\mathcal{B}_- = \{\ell \in \mathcal{B}:
\ell < 0\}$ and $\mathcal{B}_+ = \{\ell \in \mathcal{B}:
\ell > 0\}$.
In addition, let $\ell_1$ be the minimum element of the set $\mathcal{B}_+$ 
and $\ell_2$ be the maximum element of $\mathcal{B}_-$.

For any $\ell\in \mathcal{B}_+$ we have
\begin{equation}
\label{eq:badel1}
\frac{\pi^2}{2^9}\cdot \frac{(\ell+1)^2}{2^{2t_0}} \cdot 2^{2(q+1)} > \pi^2 / 32 
\Leftrightarrow  (\ell+1)^2 2^{2q} > \frac{2^7 \cdot 2^{2t_0}\pi^2}{2^5\pi^2} = 
4\cdot 2^{2t_0}.
\end{equation}
Similarly, for any $\ell \in \mathcal{B}_-$ we have
\begin{equation}
\label{eq:badel2}
\frac{\pi^2}{2^9}\cdot \frac{\ell^2}{2^{2t_0}} \cdot 2^{2(q+1)} > \pi^2 / 32 
\Leftrightarrow \ell^2 2^{2q} > \frac{2^7 \cdot 2^{2t_0}\pi^2}{2^5\pi^2} = 
4\cdot 2^{2t_0}.
\end{equation}
From \cite[Thm. 11]{BHMT02} and for $\ell \in \mathcal{B}_+$ we have
\begin{equation}
\label{eq:upperalphaplus}
|\al(m_\ell,\phi_0)|^2 = M_2^{-2} \cdot \frac{\sin^2 ((\ell + Y \cdot 2^{-p})\pi)}{\sin^2 \left(\frac{\ell
+ Y \cdot 2^{-q}}{M_2}\cdot \pi\right)}
\leq  \frac{(Y \cdot 2^{-q}\cdot \pi)^2}{\left(\frac{2\sqrt 2}{\pi} (\ell + Y\cdot 2^{-q})\pi\right)^2} \leq  \frac{\pi^2}{2^{2q}\cdot \ell^2\cdot 2^3}
\end{equation}
since $Y \cdot 2^{-q} < 1/4$. Similarly for $\ell \in \mathcal{B}_-$ we have
\begin{equation}
\label{eq:upperalphaminus}
|\al(m_\ell,\phi_0)|^2 \leq  \frac{(Y \cdot 2^{-q}\cdot \pi)^2}{\left(\frac{2\sqrt 2}{\pi} (\ell + Y\cdot 2^{-q})\pi\right)^2} \leq  \frac{\pi^2}{2^{2q}\cdot (\ell+1/4)^2\cdot 2^3}
\end{equation}
Let $P_1(\mathcal{B})$ the probability of getting a result $m_\ell \in \mathcal{B}$.
We have
\begin{eqnarray*}
P_1(\mathcal{B}) &= & \sum_{\ell \in \mathcal{B}} \sum_{j=0}^{n^d-1}|c_j|^2 |\al (m_\ell, \phi_j)|^2 \\
&= & \sum_{\ell \in \mathcal{B}_-} \sum_{j=0}^{n^d-1}|c_j|^2 |\al (m_\ell, \phi_j)|^2 + 
\sum_{\ell \in \mathcal{B}_+} \sum_{j=0}^{n^d-1}|c_j|^2 |\al (m_\ell, \phi_j)|^2
\end{eqnarray*}
We can write 
\[
\sum_{\ell \in \mathcal{B}_-} \sum_{j=0}^{n^d-1}|c_j|^2 |\al (m_\ell, \phi_j)|^2 =
\sum_{\ell \in \mathcal{B}_-} \left( |c_0|^2 |\al(m_\ell, \phi_0)|^2 + \sum_{j=1}^{n^d-1} 
|c_j|^2 |\al (m_\ell , \phi_j)|^2 \right)
\]
and
\[
\sum_{\ell \in \mathcal{B}_+} \sum_{j=0}^{n^d-1}|c_j|^2 |\al (m_\ell, \phi_j)|^2 =
\sum_{\ell \in \mathcal{B}_+} \left( |c_0|^2 |\al(m_\ell, \phi_0)|^2 + \sum_{j=1}^{n^d-1} 
|c_j|^2 |\al (m_\ell , \phi_j)|^2 \right)
\]
Note that $\sum_{j=1}^{n^d-1} |c_j|^2 = 1 - |c_0|^2 \leq 1-\frac{\pi^2}{16}$ according to the 
Lemma's assumptions, 
and $|\al(m_\ell , \phi_j)|^2 \leq 2^{-(2t_0 + 6)}$ from (\ref{eq:far}). Using the
bound from (\ref{eq:upperalphaminus}) we have
\begin{eqnarray*}
\sum_{\ell \in \mathcal{B}_-} \sum_{j=0}^{n^d-1}|c_j|^2 |\al (m_\ell, \phi_j)|^2 &\leq & \sum_{\ell \in \mathcal{B}_-} \left(|\al (m_\ell, \phi_0)|^2 
+ \frac{1-\pi^2/16}{2^{2t_0 + 6}}\right) \\ 
&\leq & 
\frac{\pi^2}{2^{2q +3}} \sum_{\ell = -2^{t_0}}^{\ell_2} \frac{1}{(\ell + 1/4)^2} + (\ell_2 + 2^{t_0}+1)
\left(1-\frac{\pi^2}{16}\right) \frac{1}{2^{2t_0+6}} \\
&\leq & \frac{\pi^2}{2^{2q +3}} \sum_{\ell = -2^{t_0}}^{\ell_2} \frac{1}{(\ell + 1/4)^2} + 
\left(1-\frac{\pi^2}{16}\right) \frac{1}{2^{t_0+6}}.
\end{eqnarray*}
We now take cases in order to calculate $\sum_{\ell = -2^{t_0}}^{\ell_2} \frac{1}{(\ell + 1/4)^2}$ depending on the value of $\ell_2$.

\textit{Case 2.1} Let $\ell_2 = -1$. Then 
\[
\sum_{\ell = -2^{t_0}}^{\ell_2} \frac{1}{(\ell + 1/4)^2} = \frac{2^4}{3^2} + \int_{-2^{t_0}}^{-1} 
\frac{1}{(x+1/4)^2} dx \leq \frac{2^4}{3^2} + \frac{4}{3} = \frac{28}{9}.
\]
As a result, 
\begin{equation}
\label{eq:prop_l2=-1}
\sum_{\ell \in \mathcal{B}_-} \sum_{j=0}^{n^d-1}|c_j|^2 |\al (m_\ell, \phi_j)|^2
\leq \frac{7}{18}\cdot \frac{\pi^2}{2^{2q}} + \left(1-\frac{\pi^2}{16}\right) \frac{1}{2^{t_0 + 6}}.
\end{equation}

\textit{Case 2.2} Let $\ell_2 < -1$. Then 
\[
\sum_{\ell = -2^{t_0}}^{\ell_2} \frac{1}{(\ell + 1/4)^2} \leq \int_{-2^{t_0}}^{\ell_2 + 1} 
\frac{1}{(x+1/4)^2} dx \leq \frac{1}{-\ell_2 - 5/4}.
\]
As a result, 
\begin{equation}
\label{eq:prop_l2<-1}
\sum_{\ell \in \mathcal{B}_-} \sum_{j=0}^{n^d-1}|c_j|^2 |\al (m_\ell, \phi_j)|^2
\leq \frac{\pi^2}{2^{2q+3}} \cdot \frac{1}{-\ell_2 - 5/4} +
\left(1-\frac{\pi^2}{16}\right) \frac{1}{2^{t_0 + 6}}.
\end{equation}

Similarly we examine the probability of the results $m_{\ell}$
for $\ell \in \mathcal{B}_+$.
Using the
bound from (\ref{eq:upperalphaplus}) we have
\begin{eqnarray*}
\sum_{\ell \in \mathcal{B}_+} \sum_{j=0}^{n^d-1}|c_j|^2 |\al (m_\ell, \phi_j)|^2 
&\leq & \sum_{\ell \in \mathcal{B}_+} \left(|\al (m_\ell, \phi_0)|^2 
+ \frac{1-\pi^2/16}{2^{2t_0 + 6}}\right) \\ 
&\leq & 
\frac{\pi^2}{2^{2q +3}} \sum_{\ell = \ell_1}^{2^{t_0}-1} \frac{1}{\ell^2} + (2^{t_0}-\ell_1 + 1)
\left(1-\frac{\pi^2}{16}\right) \frac{1}{2^{2t_0+6}} \\
&\leq & \frac{\pi^2}{2^{2q +3}} \sum_{\ell = \ell_1}^{2^{t_0}-1} \frac{1}{\ell^2} + 
\left(1-\frac{\pi^2}{16}\right) \frac{1}{2^{t_0+6}}.
\end{eqnarray*}
We now consider different values of $\ell_1$, to calculate $\sum_{\ell = \ell_1}^{2^{t_0}-1} 
\frac{1}{\ell^2}$.

\textit{Case 2.3} Let $\ell_1 = 1$. Then 
\[
\sum_{\ell = \ell_1}^{2^{t_0}-1} 
\frac{1}{\ell^2} = 1 + \int_{1}^{2^{t_0}-1} 
\frac{1}{x^2} dx \leq 2.
\]
As a result, 
\begin{equation}
\label{eq:prop_l1=1}
\sum_{\ell \in \mathcal{B}_+} \sum_{j=0}^{n^d-1}|c_j|^2 |\al (m_\ell, \phi_j)|^2
\leq \frac{\pi^2}{2^{2q+2}} + \left(1-\frac{\pi^2}{16}\right) \frac{1}{2^{t_0 + 6}}.
\end{equation}

\textit{Case 2.4} Let $\ell_1 > 1$. Then 
\[
\sum_{\ell = \ell_1}^{2^{t_0}-1} \frac{1}{\ell^2} \leq \int_{\ell_1-1}^{2^{t_0}-1} 
\frac{1}{x^2} dx \leq \frac{1}{\ell_1 - 1}.
\]
As a result, 
\begin{equation}
\label{eq:prop_l1>1}
\sum_{\ell \in \mathcal{B}_-} \sum_{j=0}^{n^d-1}|c_j|^2 |\al (m_\ell, \phi_j)|^2
\leq \frac{\pi^2}{2^{2q+3}} \cdot \frac{1}{\ell_1 - 1} +
\left(1-\frac{\pi^2}{16}\right) \frac{1}{2^{t_0 + 6}}.
\end{equation}

Let $\ell_1 = 1, \ell_2 = -1$. Then from (\ref{eq:prop_l2=-1}),(\ref{eq:prop_l1=1})
\[P_1 (\mathcal{B}) \leq \left( \frac{7}{18} + \frac{1}{4} \right) \frac{\pi^2}{2^{2q}}
+ \frac{1-\pi^2/16}{2^5} \cdot \frac{1}{2^{t_0}}.\]
From (\ref{eq:badel1}),(\ref{eq:badel2}) we have $2^{2q} > 2^{2t_0 + 2}$ 
and $2^{2q} > 2^{2t_0}$. Hence
\begin{equation}
\label{eq:l1=1,l2=-1}
P_1 (\mathcal{B}) \leq \left( \frac{7}{18} + \frac{1}{4} \right) \frac{\pi^2}{2^{2t_0 + 2}}
+ \frac{1-\pi^2/16}{2^5} \cdot \frac{1}{2^{t_0}} \leq \frac{1-\pi^2 / 16}{2^4} 2^{-t_0},
\end{equation} 
for $t_0$ sufficiently large.

Let $\ell_1 = 1, \ell_2 < -1$. From (\ref{eq:prop_l2<-1}),(\ref{eq:prop_l1=1})
\[P_1 (\mathcal{B}) \leq \frac{\pi^2}{2^{2q+3}} \cdot \frac{1}{-\ell_2 - 5/4} + 
\frac{\pi^2}{2^{2q+2}} + \frac{1-\pi^2/16}{2^5} \cdot \frac{1}{2^{t_0}}.\]
From (\ref{eq:badel1}),(\ref{eq:badel2}) we have $2^{2q} > 2^{2t_0}$ 
and $2^{2q} > \frac{2^{2t_0+2}}{\ell_2^2}$. Hence
\begin{eqnarray}
\label{eq:l1=1,l2<-1}
P_1 (\mathcal{B}) &\leq &  \frac{\pi^2}{2^3} \cdot \frac{\ell_2^2}{(-\ell_2 - 5/4)^2} \cdot
2^{-(2t_0 + 2)}
+ \frac{1-\pi^2/16}{2^5} \cdot \frac{1}{2^{t_0}} + \frac{\pi^2}{2^{2t_0+2}} \nonumber \\
&\leq & \left(\frac{\pi^2}{2^3} + \frac{1-\pi^2/16}{2^5}\right) 2^{-t_0},
\end{eqnarray} 
since $ \frac{\ell_2^2}{-\ell_2 - 5/4} \leq 2^{t_0 + 1}$ for $t_0\ge 2$.

Let $\ell_1 > 1, \ell_2 = -1$. From (\ref{eq:prop_l2=-1}),(\ref{eq:prop_l1>1})
\[P_1 (\mathcal{B}) \leq \frac{7}{18} \cdot \frac{\pi^2}{2^{2q}} 
+ \frac{\pi^2}{2^{2q+3}}\cdot \frac{1}{\ell_1-1}
+ \frac{1-\pi^2/16}{2^5} \cdot \frac{1}{2^{t_0}}.\]
From (\ref{eq:badel1}),(\ref{eq:badel2}) we have $2^{2q} > 2^{2t_0+2}$ 
and $2^{2q} > \frac{2^{2t_0+2}}{(\ell_1^2+1)^2}$. Hence
\begin{eqnarray}
\label{eq:l1>1,l2=-1}
P_1 (\mathcal{B}) &\leq &  \frac{7}{18} \cdot \frac{\pi^2}{2^{2t_0+2}}
+ \frac{1-\pi^2/16}{2^5} \cdot \frac{1}{2^{t_0}} + \frac{\pi^2}{2^{2t_0+2}}
+ \frac{\pi^2}{2^{2t_0 + 5}} \cdot \frac{(\ell_1 + 1)^2}{\ell_1 - 1} \nonumber \\
&\leq & \left(\frac{\pi^2}{2^3} + \frac{1-\pi^2 / 16}{2^5}  \right) 2^{-t_0},
\end{eqnarray} 
since $ \frac{(\ell_1+1)^2}{\ell_1 - 1} \leq 3 \cdot 2^{t_0}$ for $t_0\ge 1$.

Let  $\ell_1 > 1, \ell_2 < -1$. From (\ref{eq:prop_l2<-1}),(\ref{eq:prop_l1>1})
\[P_1 (\mathcal{B}) \leq 
\frac{\pi^2}{2^{2q+3}}\cdot \frac{1}{-\ell_2-5/4}
+ \frac{1-\pi^2/16}{2^5} \cdot \frac{1}{2^{t_0}} + \frac{\pi^2}{2^{2q+3}}\cdot \frac{1}{\ell_1 -1}.\]
From (\ref{eq:badel1}),(\ref{eq:badel2}) we have $2^{2q} > \frac{2^{2t_0+2}}{(\ell_1 + 1)^2}$ 
and $2^{2q} > \frac{2^{2t_0+2}}{\ell_2^2}$. Hence
\begin{eqnarray}
\label{eq:l1>1,l2<-1}
P_1 (\mathcal{B}) &\leq &  \frac{\pi^2}{2^{2t_0+5}} \cdot \frac{(\ell_1+1)^2}{\ell_1-1}
+ \frac{1-\pi^2/16}{2^5} \cdot \frac{1}{2^{t_0}} + 
\frac{\pi^2}{2^{2t_0+5}}\cdot \frac{\ell^2}{-\ell_2-5/4} \nonumber \\
&\leq & \left(\frac{5\pi^2}{2^5} + \frac{1-\pi^2/16}{2^5} \right) 2^{-t_0}.
\end{eqnarray} 

Finally, combining the results from 
(\ref{eq:l1=1,l2=-1}),(\ref{eq:l1=1,l2<-1}),(\ref{eq:l1>1,l2=-1}) and (\ref{eq:l1>1,l2<-1}) we have
\[
P_1(\mathcal{B}) \leq \left(\frac{5\pi^2}{2^5} + \frac{1-\pi^2/16}{2^5} \right) 2^{-t_0}.
\] 
\end{proof}


\begin{thebibliography}{9}

\bibitem{AL99} D. S. Abrams and S. Lloyd, Quantum algorithm providing exponential speed increase
for nding eigenvalues and eigenvectors, Phys. Rev. Lett., 83 (1999), 5162-5165.

\bibitem{AC11} B. Andrews and J. Clutterbuck, Proof of the fundamental gap conjecture, J. Amer.
Math. Soc., 24 (2011), 899-916.

\bibitem{BHMT02} G. Brassard, P. Hoyer, M. Mosca  and A. Tapp, Quantum amplitude amplification and estimation,
Quantum Computation and Quantum Information: A Millenium Volume, AMS Contemporary Mathematics Series, (2002).
Also quant-ph/0005055.

\bibitem{Grover97} L. K. Grover, Quantum Mechanics Helps in Searching for a Needle in a Haystack,
Phys. Rev. Lett., 79 (1997), 325-328. 

\bibitem{KKR06} J. Kempe, A. Kitaev, and O. Regev, The complexity of the Local Hamiltonian Problem,
SIAM J. Computing, 35 (2006), 1070-1097.

\bibitem{KR01} A. Klappenecker and M. R\"otteler, Discrete cosine transforms on quantum computers,
In Proceedings of the 2nd International Symposium on Image and Signal Processing and
Analysis, (2001), 464-468.

\bibitem{Lub} C. Lubich, From Quantum to Classical Molecular Dynamics: Reduced Models and Numerical
Analysis, European Mathematical Society, Z\"urich, 2008.

\bibitem{NC00} M. A. Nielsen and I. Chuang, Quantum Computation and Quantum Information, Cambridge
University Press, 2000.

\bibitem{Pap07} A. Papageorgiou, On the complexity of the multivariate Sturm-Liouville eigenvalue
problem, J. Complex., 23 (2007), 802-827.

\bibitem{PPTZ12} A. Papageorgiou, I. Petras, J.F. Traub, and C. Zhang, A fast algorithm for approximating
the ground state energy on a quantum computer, Math. Comp., 82 (2013), 2293-2304.

\bibitem{PZ12} A. Papageorgiou and C. Zhang, On the efficiency of quantum algorithms for Hamiltonian
simulation, Quantum Information Processing, 11 (2012), 541-561.

\bibitem{Shor97} P. W. Shor, Polynomial-time algorithms for prime factorization and discrete logarithms
on a quantum computer, SIAM J. Comput., 26 (1997), 1484-1509.

\bibitem{Suzuki90} M. Suzuki, Fractal decomposition of exponential operators with applications to manybody
theories and Monte Carlo simulations, Physics Letters A, 146:6 (1990), 319-323.

\bibitem{Suzuki92} M. Suzuki, Fractal path integrals with applications to quantum many-body systems,
Physica A: Statistical Mechanics and its Applications, 191:1-4 (1992), 501-515.

\bibitem{TM01} B. C. Travaglione and G. J. Milburn, Generation of eigenstates using the phase estimation
algorithm, Phys. Rev. A, 63:032301, 2001.

\bibitem{Weinberger56} H. F. Weinberger, Upper and lower bounds for eigenvalues by finite difference methods,
Communications on Pure and Applied Mathematics, 9:3 (1956), 613-623.

\bibitem{Weinberger58} H. F. Weinberger, Lower bounds for higher eigenvalues by finite difference methods,
Pacific J. Math., 8 (1958), 339-368.

\bibitem{Wick94} M. V. Wickerhauser, Adapted wavelet analysis from theory to software, A.K. Peters,
Wellesley, MA, 1994.

\end{thebibliography}

\end{document}